\documentclass[12pt]{article}
\usepackage{graphpap, graphicx, float, tabularx, array, setspace, natbib, multirow}
\usepackage{amsmath, amsthm, amsfonts, eucal}
\usepackage[multiple]{footmisc}
\graphicspath{{../output/}}

\long\def\symbolfootnote[#1]#2{\begingroup
\def\thefootnote{\fnsymbol{footnote}}\footnote[#1]{#2}\endgroup}
\numberwithin{equation}{section}
\newenvironment{proofThm}[1][Proof of Theorem \ref{relValueThm}.]{\begin{trivlist}
\item[\hskip \labelsep {\bfseries #1}]}{\hfill\qed\end{trivlist}}
\theoremstyle{plain}
\newtheorem{thm}{Theorem}[section]
\newtheorem{cor}[thm]{Corollary}
\newtheorem{ass}[thm]{Assumption}
\renewcommand{\baselinestretch}{1.3}

\newcommand\blfootnote[1]{%
  \begingroup
  \renewcommand\thefootnote{}\footnote{#1}%
  \addtocounter{footnote}{-1}%
  \endgroup
}

\def\F{\CMcal{F}} 
\def\O{\Omega}
\def\o{\omega}
\def\L{\Lambda}
\def\p{\pi}
\def\R{{\mathbb R}}
\def\ito{It{\^o}}
\def\intt{\int_0^t}
\def\sgn{{\rm sgn}}
\def\ep{\epsilon}

\begin{document}

\centerline{\bf \LARGE{The Rank Effect}}



\blfootnote{The views expressed in this paper are those of the authors alone and do not necessarily reflect the views of the Federal Reserve Bank of Dallas or the Federal Reserve System. All remaining errors are our own.}

\vskip 30pt

\centerline{\large{Ricardo T. Fernholz\symbolfootnote[2]{Claremont McKenna College, 500 E. Ninth St., Claremont, CA 91711, rfernholz@cmc.edu} \hskip 90pt Christoffer Koch\symbolfootnote[3]{Federal Reserve Bank of Dallas, 2200 North Pearl Street, Dallas, TX 75201, christoffer.koch@dal.frb.org}}}

\vskip 3pt

\centerline{\hskip 17pt Claremont McKenna College \hskip 57pt Federal Reserve Bank of Dallas}

\vskip 35pt


\vskip 5pt

\centerline{\large{\today}}

\vskip 50pt

\renewcommand{\baselinestretch}{1.1}
\begin{abstract}
We decompose returns for portfolios of bottom-ranked, lower-priced assets relative to the market into rank crossovers and changes in the relative price of those bottom-ranked assets. This decomposition is general and consistent with virtually any asset pricing model. Crossovers measure changes in rank and are smoothly increasing over time, while return fluctuations are driven by volatile relative price changes. Our results imply that in a closed, dividend-free market in which the relative price of bottom-ranked assets is approximately constant, a portfolio of those bottom-ranked assets will outperform the market portfolio over time. We show that bottom-ranked relative commodity futures prices have increased only slightly, and confirm the existence of substantial excess returns predicted by our theory. If these excess returns did not exist, then top-ranked relative prices would have had to be much higher in 2018 than those actually observed --- this would imply a radically different commodity price distribution.
\end{abstract}
\renewcommand{\baselinestretch}{1.3}

\vskip 50pt

JEL codes: G11, G12, G14, C14

Keywords: Asset pricing, asset pricing anomalies, asset pricing factors, efficient markets, commodity prices, nonparametric methods

\vfill

\pagebreak

\section{Introduction}

Consider a market consisting of multiple assets. Suppose that we divide those assets into a top-ranked half consisting of the most expensive assets and a bottom-ranked half consisting of the least expensive assets. Two facts emerge. First, in the absence of dividends, the return for the bottom-ranked half of assets relative to the top-ranked half will depend on two factors: the change in the price of bottom-ranked assets relative to top-ranked assets and the number of assets that crossover from bottom half to top half. Second, if the relative price of bottom-ranked versus top-ranked assets is constant over time, then the crossovers will ensure that the bottom-ranked half of assets deliver a higher return than the top-ranked half of assets --- a \emph{rank effect}. These two facts amount to accounting identities, and we explore their implications.

Using novel mathematical methods in which we represent asset prices as general continuous semimartingales, we show that returns for portfolios of ranked asset subsets relative to the market can be decomposed into rank crossovers and changes in the relative price of ranked asset subsets. Specifically, we characterize relative returns as
\begin{equation} \label{intuitiveEq}
 \text{relative return} \;\; = \;\; \text{rank crossovers} \; +  \begin{array}{c} \text{change in relative price of} \\[-0.2cm] \text{ranked asset subsets} \end{array}, 
\end{equation}
where the relative prices of ranked asset subsets are volatile and drive relative return fluctuations. Rank crossovers, in contrast, are roughly constant, and non-negative for portfolios of bottom-ranked, lower-priced assets and non-positive for portfolios of top-ranked, higher-priced assets. The decomposition \eqref{intuitiveEq} is achieved using few assumptions about the dynamics of individual asset prices, which means our results are sufficiently general that they apply to almost any equilibrium asset pricing model, in a sense made precise in Section \ref{sec:theory} below. Indeed, the decomposition \eqref{intuitiveEq} is little more than an accounting identity that is approximate in discrete time and exact in continuous time.

This generality is one of the strengths of our unconventional approach. The continuous semimartingales we use to represent asset prices allow for a practically unrestricted structure of time-varying dynamics and co-movements that are consistent with the endogenous price dynamics of any economic model. In this manner, our framework is applicable to both rational \citep{Sharpe:1964,Lucas:1978,Cochrane:2005} and behavioral \citep{Shiller:1981,DeBondt/Thaler:1989} theories of asset prices, as well as to the many econometric specifications of asset pricing factors identified in the empirical literature.

Our results have implications for the efficiency of asset markets. We show that in a market in which dividends and the entry and exit of assets over time play negligible roles, the non-negativity of rank crossovers in \eqref{intuitiveEq} implies that portfolios of bottom-ranked, lower-priced assets must necessarily outperform the market except in the special case where the relative price of those bottom-ranked assets falls at a sufficiently fast rate over time. Thus, in order to rule out rank effect predictable excess returns, the relative price of lower-priced assets must be decreasing on average over time at a rate fast enough to overwhelm the predictable positive rank crossovers. This result re-casts market efficiency in terms of a constraint on the dynamic behavior of ranked relative asset prices. Through this lens, our approach provides a novel mechanism to uncover risk factors or inefficiencies that persist across a variety of different asset markets.

The general mathematical approach we use to derive the decomposition \eqref{intuitiveEq} is similar to \citet{Fernholz/Stroup:2018}. Like these authors, we are able to characterize a class of asset pricing factors that are universal across different economic models and econometric specifications by tying the volatility of relative portfolio returns to changes in the relative price of ranked asset subsets. In this way, the decomposition \eqref{intuitiveEq} is naturally immune to recent criticisms of the implausibly high number of factors and anomalies uncovered by the empirical asset pricing literature \citep{Novy-Marx:2014,Harvey/Liu/Zhu:2016,Bryzgalova:2016}. 

We test our theoretical predictions using commodity futures. This market provides a clear test of our theory, since commodity futures contracts do not pay dividends and rarely exit from the market. Although dividends and asset entry/exit can be incorporated into our framework and do not overturn the basic insight of the decomposition \eqref{intuitiveEq}, they do alter and complicate the form of our results. Our empirical analysis focuses on simple price-weighted portfolios of bottom- and top-ranked subsets of lower- and higher-priced commodities. We decompose the relative returns of these bottom- and top-ranked commodity futures portfolios from 1974-2018 as in \eqref{intuitiveEq}, with the market portfolio defined as the price-weighted portfolio that holds one unit of each commodity futures contract. Empirically, we show that the relative price of bottom-ranked commodity futures declined only slightly over the forty year period we study. Consequently, and as predicted by the theory, the bottom-ranked portfolio exhibits positive long-run returns relative to the market driven by accumulating positive rank crossovers, while the top-ranked portfolio exhibits negative long-run returns relative to the market driven by accumulating negative rank crossovers. The bottom-ranked portfolio of lower-priced commodity futures consistently and substantially outperforms the price-weighted market portfolio, with excess returns that have Sharpe ratios of 0.6-0.8 in most decades.

This rank effect is similar to the value anomaly for commodity futures described by \citet{Asness/Moskowitz/Pedersen:2013}. These authors rank commodities based on their current price relative to their past prices, and show that portfolios consisting of high-value commodity futures --- those contracts with low prices today relative to the past --- consistently and substantially outperform portfolios consisting of low-value commodity futures --- those contracts with high prices today relative to the past. Given its similarity to our results, the value anomaly for commodity futures can thus be interpreted in terms of the decomposition \eqref{intuitiveEq}. This decomposition implies that there is no rank effect for commodities, and hence no value anomaly for commodities as well, only if the relative price of bottom-ranked commodities consistently and rapidly declines. In order to investigate this possibility, we perform counterfactuals in which we impose such a rapid relative price decline. We find that any scenario in which there is no rank effect or value anomaly implies a radically different set of relative commodity futures prices in 2018. In particular, bottom- and top-ranked commodity prices would have had to be orders of magnitude farther apart than they were in 2018. This finding yields a novel interpretation of the results of \citet{Asness/Moskowitz/Pedersen:2013}.

It is important to emphasize that the mathematical methods we use to derive our results are well-established and the subject of active research in statistics and mathematical finance. Our results and methods are most similar to \citet{Karatzas/Ruf:2017}, who provide a return decomposition similar to \eqref{intuitiveEq}. Their results are based on the original characterization of \citet{Fernholz:2002}, which led to subsequent contributions by \citet{Banner/Fernholz/Karatzas:2005}, \citet{Pal/Pitman:2008}, and \citet{Fernholz:2017a}, among others. These contributions focus primarily on the solutions of stochastic differential equations and the link between rank crossovers and the relative growth rates and values of different ranked processes. Our results, in contrast, are interpreted in an economic setting that focuses on questions of asset pricing risk factors and market efficiency. Ours is also the first paper to empirically examine these results in the commodity futures market, which, as discussed above, most closely aligns with the assumptions that underlie our theoretical results.

Our results raise questions regarding the implications of equilibrium asset pricing models for relative prices of different ranked assets. Since our results show that the relative prices of different ranked asset subsets operate as a universal risk factor, different models' predictions for these relative prices becomes a major question of interest. In particular, our results imply that ranked relative prices should be linked to an endogenous stochastic discount factor that in equilibrium is linked to the marginal utility of economic agents. It is not obvious what economic and financial forces might underlie such a link, however. Nonetheless, our paper shows that unless the relative price of bottom-ranked, lower-priced assets is consistently and rapidly falling over time, such links must necessarily exist.

\vskip 50pt

\section{A Simple Example} \label{sec:example}

Suppose that the economy consists of $N > 1$ assets whose \emph{individual} prices at time $t$ are denoted by $p_i(t)$, $i = 1, \ldots, N$. We wish to rank these individual asset prices from most to least expensive. Let $p_{(k)}(t)$ denote the price of the $k$-th \emph{ranked} (most expensive) asset at time $t$, so that
\begin{equation}
 p_{(1)}(t) \geq p_{(2)}(t) \geq \cdots \geq p_{(N)}(t).
\end{equation}
We rank assets based on their individual prices relative to each other. The distinction between individual and ranked assets is important. Individual assets occupy different ranks over time since their prices change relative to each other in response to various shocks. A ranked asset refers to the individual asset that occupies a specific rank at a moment in time.


For any $1 \leq m < N$, denote the top $m$ ranked assets at time $t$ by
\begin{equation} \label{topSub}
 P_{\text{top}}(t) = \left\{p_{(1)}(t), \ldots, p_{(m)}(t)\right\},
\end{equation}
and the bottom $N-m$ ranked assets at time $t$ by
\begin{equation} \label{botSub}
 P_{\text{bot}}(t) = \left\{p_{(m+1)}(t), \ldots, p_{(N)}(t)\right\}.
\end{equation}
From period $t$ to period $t+1$, each of the \emph{individual} assets in the top-ranked subset \eqref{topSub} can either stay in that subset or crossover into the bottom-ranked subset \eqref{botSub}. More precisely, for each individual $p_i(t) \in P_{\text{top}}(t)$, either $p_i(t + 1) \in P_{\text{top}}(t+1)$ or $p_i(t + 1) \in P_{\text{bot}}(t+1)$. Similarly, each of the \emph{individual} assets in the bottom-ranked subset in period $t$ can either stay in that subset or crossover into the top-ranked subset. More precisely, for each individual $p_i(t) \in P_{\text{bot}}(t)$, either $p_i(t + 1) \in P_{\text{bot}}(t+1)$ or $p_i(t + 1) \in P_{\text{top}}(t+1)$.


If the bottom- and top-ranked asset subsets grow at similar rates over time, then rank crossovers between the two subsets ensure a higher return for the bottom-ranked assets relative to the top-ranked assets on average over time. To see why, note that the individual assets that drop out of the top subset grow more slowly than the individual assets that rise out of the bottom subset. This simple observation implies that the individual assets in the top subset must grow more slowly than the individual assets in the bottom subset on average over time. Thus, if the two ranked asset subsets grow at similar rates over time, then rank crossovers ensure that the individual assets in the bottom-ranked subset generate higher returns than the individual assets in the top-ranked subset --- a rank effect emerges.

This argument is depicted in Figure \ref{simpleFig}. On average, the prices of the top- and bottom-ranked subsets remain stable relative to each other from time $t$ to time $t+1$ (arrows a and b). However, because some individual assets crossover from the top subset to the bottom subset (arrow c) while some individual assets crossover from the bottom subset to the top subset (arrow d) over this same period, the return of the individual assets in the bottom subset from time $t$ to time $t+1$ will be higher on average than the return of the individual assets in the top subset from time $t$ to time $t+1$. Although this intuitive argument is purposely illustrative and informal, we demonstrate its validity in a rigorous framework in the next section.

Of course, this intuitive argument for the rank effect relies on stability in the prices of bottom-ranked assets relative to top-ranked assets over time. If top-ranked assets consistently rise in price relative to bottom-ranked assets, then the rank effect may disappear. Consistent relative price gains are unlikely to last for long time periods, however, since these gains would eventually lead to the bottom-ranked assets disappearing relative to the top ranks. Thus, while the relative prices of different ranked asset subsets may fluctuate substantially over time, it seems less likely that these relative prices will consistently trend up or down over long periods. In Section \ref{sec:empirics}, we confirm that the relative prices of ranked asset subsets of commodity futures from 1974-2018 are approximately stable in this way, as expected.

The possibility that assets may enter and exit the market over time presents another important caveat. The intuitive argument presented in this section relies on a fixed set of the same $N$ assets that do not change over time. In the real world, however, assets may disappear from a market for a variety of reasons including bankruptcy and loss of interest from investors. If the rate of asset exit from the market is sufficiently large, then this may lower the returns of lower-ranked assets --- those that are most likely to exit --- enough to offset the rank effect. Ultimately, whether or not asset entry and exit is sufficiently impactful to negate the rank effect is an empirical question that depends on the specific market being considered. We choose to focus on the commodity futures market for our empirical analysis in Section \ref{sec:empirics} in part because assets rarely exit from this market.

Finally, the simple logic behind the rank effect depicted in Figure \ref{simpleFig} is most salient for markets of assets that do not pay dividends. In the absence of dividends, price change capital gains entirely determine returns, which is central to this section's intuitive explanation. In the presence of dividends, it is possible that bottom-ranked assets pay lower dividends than top-ranked assets. If this differential rate of dividend payment is large enough, then the rank effect may disappear even though the earlier argument about rank crossovers and relative prices remains valid. In such a scenario, bottom-ranked assets would still generate higher capital gains than top-ranked assets, as depicted in Figure \ref{simpleFig}, but this differential rate of capital gains would be negated by the relatively higher dividends paid by top-ranked assets. In this case, there is a rank effect for capital gains, but different dividend rates imply no rank effect for returns.

In the proceeding sections, we show that the the simple logic depicted in Figure \ref{simpleFig} holds in a rigorous mathematical framework. We then examine commodity futures prices and show that from 1974-2018 there has been a large and consistent rank effect in which portfolios of bottom-ranked assets outperform portfolios of top-ranked assets. In line with the preceding discussion, the decision to focus on commodity futures is primarily motivated by the fact that these assets do not pay dividends and only rarely exit from the market. Indeed, no such exits occur over the 1974-2018 time period we consider. In this way, the choice of commodity futures aligns our empirical analysis as closely as possible with both the simple argument we describe in this section and the rigorous theory we develop in the next section.




\vskip 50pt

\section{Theory} \label{sec:theory}

In this section, we characterize the relationship between the relative prices of bottom- and top-ranked asset subsets, rank crossovers, and the returns for portfolios of ranked assets relative to the market. This is accomplished using the continuous semimartingale asset price representation of \citet{Fernholz/Stroup:2018}. These general methods allow for a relative return characterization that is broad and nests virtually all equilibrium asset pricing theories, so that our results require no commitment to specific models of trading behavior, agent beliefs, or market microstructure.

\subsection{Setup} \label{sec:setup}
Consider a market that consists of $N > 1$ assets. Time is continuous and denoted by $t$ and uncertainty in this market is represented by a probability space $(\O, \F, P)$ that is endowed with a right-continuous filtration $\{\F_t ; t \geq 0\}$. Each asset price $p_i$, $i = 1, \ldots, N$, is represented by a positive continuous semimartingale that is adapted to $\{\F_t ; t \geq 0\}$, so that
\begin{equation} \label{contSemimart}
 p_i(t) = p_i(0) + g_i(t) + v_i(t),
\end{equation}
where $g_i$ is a continuous process of finite variation, $v_i$ is a continuous, square-integrable local martingale, and $p_i(0)$ is the initial price. The semimartingale representation \eqref{contSemimart} decomposes asset price dynamics into a time-varying cumulative growth component, $g_i(t)$, whose total variation is finite over every interval $[0, T]$, and a randomly fluctuating local martingale component, $v_i(t)$. By representing asset prices as continuous semimartingales, we are able to impose almost no structure on the underlying economic environment.


Our approach in this paper is unconventional in that we do not commit to a specific model of asset pricing. Instead, following \citet{Fernholz/Stroup:2018}, we derive results in a very general setting, with the understanding that the minimal assumptions behind these results mean that they will be consistent with almost any underlying economic model. Indeed, essentially any asset price dynamics generated endogenously by a model can be represented as general continuous semimartingales of the form \eqref{contSemimart}. This generality is crucial, since our goal is to derive results that apply to all economic and financial environments.

What are the assumptions the framework \eqref{contSemimart} relies on, and how do these assumptions relate to other asset pricing theories? A first important assumption is that assets do not pay dividends, so that returns are driven entirely by capital gains via price changes. In this sense, we can think of the $N$ assets in the market as rolled-over futures contracts that guarantee delivery of some underlying real asset on a future date. We emphasize that this assumption is for simplicity only. Our results can easily be extended to include general continuous semimartingale dividend processes similar to \eqref{contSemimart}. Including such dividend processes complicates the theory but does not change the basic insight of our results, a point that we discuss further below.

Second, we consider a closed market in which there is no asset entry or exit over time. We assume that the $N$ assets in the market are unchanged over time. In Appendix \ref{sec:entry}, we extend our main results to a market in which assets occasionally exit and are replaced by new assets over time. While such entry and exit does complicate our results, it does not alter the basic insight that rank crossovers and changes in the relative prices of ranked asset subsets are key determinants of returns as illustrated by \eqref{intuitiveEq}. Furthermore, for our empirical analysis in Section \ref{sec:empirics} we consider commodity futures contracts in which asset exit --- the more significant omission --- does not occur over our sample period, thus aligning our empirical analysis as closely as possible with the theoretical assumptions of no dividends and no entry or exit.

Another assumption behind \eqref{contSemimart} is that the prices $p_i$ are continuous functions of time $t$ that are adapted to the filtration $\{\F_t ; t \geq 0\}$. The assumption of continuity is essential for mathematical tractability, since we rely on stochastic differential equations whose solutions are readily obtainable in continuous time to derive our theoretical results. Given the generality of our setup, it is difficult to see how introducing instantaneous jumps into the asset price dynamics \eqref{contSemimart} would meaningfully alter our conclusions. Nonetheless, it is important and reassuring that in Section \ref{sec:empirics} we confirm the validity of our continuous-time results using monthly, discrete-time asset price data. This is not surprising, however, since an instantaneous price jump is indistinguishable from a rapid but continuous price change --- this is allowed according to \eqref{contSemimart} --- using discrete-time data. Finally, the assumption that asset prices $p_i$ are adapted means only that they cannot depend on the future. This reflects the reality that agents cannot relay information about the future realization of stochastic processes to the present.

The decomposition \eqref{contSemimart} separates asset price dynamics into two distinct parts. The first, the finite variation process $g_i$, has an instantaneous variance of zero (zero quadratic variation) and measures the cumulative growth in price over time. Despite its finite variation, the cumulative growth process $g_i$ can constantly change depending on economic and financial conditions as well as other factors, including the prices of the different assets. In Section \ref{sec:results}, we show that our main relative return decomposition result consists of a finite variation process as well. In the subsequent empirical analysis Section \ref{sec:empirics}, we construct this finite variation process using discrete-time asset price data and show a clear contrast between its time-series behavior and the behavior of processes with positive quadratic variation (instantaneous variance greater than zero). 

The second part of the decomposition \eqref{contSemimart} consists of the square-integrable local martingale $v_i$. In general, this process has a positive instantaneous variance (positive quadratic variation), and thus its fluctuations are much larger and faster than for the finite variation cumulative growth process $g_i$. Note that a local martingale is more general than a martingale \citep{Karatzas/Shreve:1991}, and thus includes an extremely broad class of continuous stochastic processes. Intuitively, the process $v_i$ can be thought of as a random walk with a variance that can constantly change depending on economic and financial conditions as well as other factors. Furthermore, we allow for a rich structure of potentially time-varying covariances among the local martingale components $v_i$ of different asset prices.

The commodity futures market we apply our theoretical results to offers one of the cleanest applications of our theory, since commodities rarely exit the market and their futures contracts do not pay dividends. A number of studies have decomposed commodity futures prices into risk premia and forecasts of future spot prices \citep{Fama/French:1987,Chinn/Coibion:2014}. Commodity spot prices, which are a major determinant of futures prices, have in turn been linked to storage costs and fluctuations in supply and demand \citep{Brennan:1958,Alquist/Coibion:2014}. In the context of this literature, there are many potential mappings from the fundamental economic and financial forces that determine spot and futures commodity prices to the general continuous semimartingale representation of asset prices \eqref{contSemimart}. Indeed, higher storage costs, increases in demand, rising risk premia, and many other factors can be represented as increases in the cumulative growth process $g_i$. Similarly, all of the unpredictable random shocks that impact commodity markets can be represented as changes in the local martingale $v_i$.

The crucial point, however, is that all of these models and the different economic and financial factors that they emphasize are consistent with the reduced form representation of asset prices \eqref{contSemimart}. After all, any model that proposes an explanation for the growth and volatility of commodity prices can be translated into our setup. The advantage of \eqref{contSemimart} is that we need not commit to any particular model of asset pricing, thus allowing us to derive results that are consistent across all the different models.

\subsection{Ranked Assets} \label{sec:rank}
In order to characterize the rank effect, it is necessary to introduce notation for ranked asset prices. Before doing so, however, it is useful to define relative asset prices. Let $\theta = (\theta_1, \ldots, \theta_N)$, where each $\theta_i$, $i = 1, \ldots, N$, is given by
\begin{equation} \label{theta}
 \theta_i(t) = \frac{p_i(t)}{\sum_{i=1}^N p_i(t)}.
\end{equation}
Because the continuous semimartingales $p_i$ are all positive by assumption, it follows that $0 < \theta_i < 1$, for all $i = 1, \ldots, N$. By construction, we also have that $\theta_1 + \cdots + \theta_N = 1$. 

For $k = 1, \ldots, N$, let $p_{(k)}(t)$ represent the price of the $k$-th most expensive asset at time $t$, so that
\begin{equation}
 \max (p_1(t), \ldots, p_N(t)) = p_{(1)}(t) \geq p_{(2)}(t) \geq \cdots \geq p_{(N)}(t) = \min (p_1(t), \ldots, p_N(t)),
\end{equation}
and let $\theta_{(k)}(t)$ be the relative price of the $k$-th most expensive asset at time $t$, so that
\begin{equation} \label{rankTheta}
\theta_{(k)}(t) = \frac{p_{(k)}(t)}{\sum_{i=1}^N p_i(t)}.
\end{equation}
In our analysis below, we consider portfolios consisting of bottom- and top-ranked subsets of the $N$ assets. For this purpose, it useful to introduce notation for the relative prices of subsets of bottom- and top-ranked assets. For any cutoff $1 \leq c < N$, let $\Theta_{sc}$ denote the relative price of the $N - c$ bottom-ranked, lowest-priced assets, so that
\begin{equation}
 \Theta_{sc}(t) = \sum_{k = c + 1}^N \theta_{(k)}(t).
\end{equation}
Similarly, let $\Theta_{bc}$ denote the relative price of the $c$ top-ranked, highest-priced assets, so that
\begin{equation}
 \Theta_{bc}(t) = \sum_{k=1}^c \theta_{(k)}(t),
\end{equation}
for all $1 \leq c \leq N$.

In order to describe the dynamics of portfolios consisting of ranked assets, it is necessary to introduce the notion of a local time process. Local time processes are necessary because the rank function is not differentiable and hence we cannot simply apply \ito's Lemma when describing the behavior of portfolios consisting of ranked assets. For any continuous process $x$, the \emph{local time} at $0$ for $x$ is the process $\L_x$ defined by
\begin{equation} \label{localTime}
 \L_x(t) = \frac{1}{2}\left(|x(t)| - |x(0)| - \intt\sgn(x(s))\,dx(s)\right).
\end{equation}
As detailed by \citet{Karatzas/Shreve:1991} and \citet{Fernholz:2017a}, the local time for $x$ measures the amount of time the process $x$ spends near zero and can also be defined as
\begin{equation}
  \L_x(t) = \lim_{\ep \downarrow 0} \frac{1}{2\ep} \int_0^t 1_{\{ |x(s)| < \ep \}}\,ds.
\end{equation}
To be able to link asset rank to asset index, let $\o_t$ be the permutation of $\{1, \ldots, N\}$ such that for $1 \leq i, k \leq N$,
\begin{equation} \label{pTK}
 \o_t(k) = i \quad \text{if} \quad p_{(k)}(t) = p_i(t).
\end{equation}
This definition implies that $\o_t(k) = i$ whenever asset $i$ is the $k$-th most expensive asset at time $t$, with ties broken in some consistent manner.\footnote{For example, if $p_i(t) = p_j(t)$ and $i > j$, then we can set $\o_t(k) = i$ and $\o_t(k+1) = j$.}

Local time processes play a key role in decomposing the returns of different portfolios of ranked assets. In particular, the local times of the processes $\theta_{(k)} - \theta_{(k+1)}$, $k = 1, \ldots, N-1$, appear throughout our analysis below. In order to simplify the notation, then, let 
\begin{equation}
 \L^k(t) = \L_{\theta_{(k)} - \theta_{(k+1)}}(t),
\end{equation}
for all $k = 1, \ldots, N-1$. The local time processes $\L^k$ measure the intensity of rank crossovers. More specifically, for each $k = 1, \ldots, N-1$, the local time process $\L^k$ measures the intensity at which the $k$-th and $k+1$-ranked assets crossover in rank, since any time the process $\theta_{(k)} - \theta_{(k+1)}$ reaches zero there is a potential crossover in rank. Finally, for notational completeness, let $\L^N = 0$. 

\begin{ass} \label{triplePoint}
For all $t$, if $p_i(t) = p_j(t)$ for some $i, j = 1, \ldots, N$, $i \neq j$, then $p_i(t) \neq p_\ell(t)$ for all $\ell = 1, \ldots, N$, $\ell \neq i$ and $\ell \neq j$.
\end{ass}

Assumption \ref{triplePoint} ensures that no triple points occur, meaning that no more than two different assets have the same price at the same time $t$. This assumption is for simplicity. Our main results in Section \ref{sec:results} can be extended to markets in which triple points do occur, but such an extension substantially complicates the exposition of these results without changing the basic insight.\footnote{In fact, Assumption \ref{triplePoint} is slightly stronger than what is needed to simplify the results in Section \ref{sec:results}. The assumption we need for such a simplification is only that the local times at zero of the processes $\theta_{(k)} - \theta_{(k+\ell)}$, $k = 1, \ldots, N-2$, $\ell \geq 2$, are equal to zero. See \citet{Ichiba/Papathanakos/Banner/Karatzas/Fernholz:2011} and \citet{Ichiba/Karatzas/Shkolnikov:2013}.}  We refer the reader to the more general framework of \citet{Karatzas/Ruf:2017} for such an extension of our results. 


\subsection{Portfolios of Ranked Assets} \label{sec:port}
A \emph{portfolio of ranked assets} $\p(t) = (\p_{1}(t), \ldots, \p_{N}(t))$ states that $\p_{k}(t)$ shares of the $k$-th ranked asset $\o_t(k)$ should be held at time $t$, for each rank $k = 1, \ldots, N$. More precisely, a portfolio of ranked assets $\p$ states that $\p_{k}$ shares of the asset with price $p_{(k)}$, which is the price of the $k$-th ranked asset, should be held. The shares $\p_{1}, \ldots, \p_{N}$ that make up a portfolio must be measurable, adapted, and non-negative.\footnote{The assumption that portfolios hold only non-negative shares of each asset, and hence do not hold short positions, is for simplicity only. Our theory and results can be extended to long-short portfolios as well.} The \emph{value} of the portfolio $\p$ is denoted by $V_{\p} > 0$, and satisfies
\begin{equation} \label{valueEq}
 V_{\p}(t) = \sum_{k=1}^N \p_{k}(t)p_{(k)}(t), 
\end{equation}
for all $t$. 

It is sometimes also useful to describe portfolios $\p$ in terms of \emph{weights}, denoted by $w^{\p}(t) = (w^{\p}_{1}(t), \ldots, w^{\p}_{N}(t))$, which measure the fraction of portfolio $\p$ invested in each asset. The shares of each ranked asset held by a portfolio, $\p_{k}$, are easily linked to the weights of that portfolio, $w^{\p}_{k}$. In particular, the portfolio $\p(t) = (\p_{1}(t), \ldots, \p_{N}(t))$ has weights equal to
\begin{equation} \label{weightsEq}
 w^{\p}_{k}(t) = \frac{p_{(k)}(t)\p_{k}(t)}{V_{\p}(t)},
\end{equation}
for all ranks $k = 1, \ldots, N$ and all $t$, since \eqref{weightsEq} is equal to the dollar value invested in the $k$-th ranked asset by portfolio $\p$ divided by the dollar value of portfolio $\p$. It is easy to confirm using \eqref{valueEq} and \eqref{weightsEq} that the weights $w^{\p}_{k}$ sum up to one.

We require that all portfolios satisfy the self-financibility constraint, which ensures that gains or losses from the portfolio $\p$ account for all changes in the value of the investment over time. This implies that
\begin{equation}
 V_{\p}(t) - V_{\p}(0) = \int_0^t \sum_{k=1}^N \p_{k}(t) \, dp_{\o_t(k)}(t),
\end{equation}
for all $t$. For comparability and without loss of generality, we set the initial holdings for all portfolios equal to the combined initial value of all assets in the economy, so that
\begin{equation} \label{valueNormalizationEq}
 V_{\p}(0) = \sum_{k=1}^N p_{(k)}(0),
\end{equation}
for all portfolios $\p$.

One simple example of a portfolio that plays a central role in much of our theoretical and empirical analysis is the \emph{market portfolio}, which we denote by $m$. The market portfolio $m$ holds one share of each asset, so that $m(t) = (1, \ldots, 1)$ for all $t$. Following \eqref{valueEq}, we have that the value of the market portfolio, $V_m$, is given by
\begin{equation} \label{marketValueEq}
  V_m(t) = \sum_{k=1}^N m_{k}(t)p_{(k)}(t) = \sum_{k=1}^N p_{(k)}(t) = \sum_{i=1}^N p_i(t),
\end{equation}
for all $t$. The last equality of \eqref{marketValueEq} follows because the sum of all $N$ asset prices is always the same, regardless of whether those prices are summed across rank $k$ or index $i$. Note that the market portfolio satisfies the self-financibility constraint, since\footnote{We use the notation $dp_{\o_t(k)}(t)$ for simplicity, and note that $dp_{\o_t(k)}(t) = \sum_{i=1}^N 1_{\{\o_t(k) = i\}}dp_i(t)$, for all $k = 1, \ldots, N$ and all $t$.}
\begin{equation}
 \int_0^t \sum_{k=1}^N m_{k}(t) \, dp_{\o_t(k)}(t) = \int_0^t \sum_{i=1}^N \, dp_i(t) = \sum_{k=1}^N p_i(t) - \sum_{i=1}^N p_i(0) = V_m(t) - V_m(0),
\end{equation}
for all $t$. It also satisfies the initial condition \eqref{valueNormalizationEq}, as shown by evaluating \eqref{marketValueEq} at $t = 0$. 

The market portfolio $m$ has weights equal to the ranked relative prices $\theta_{(k)}$ defined in \eqref{rankTheta}. This follows from \eqref{weightsEq} and \eqref{marketValueEq}, which together imply that
\begin{equation} \label{marketWeights}
 w^m_{k}(t) = \frac{p_{(k)}(t)}{V_m(t)} = \frac{p_{(k)}(t)}{\sum_{i=1}^N p_i(t)} = \theta_{(k)}(t),
\end{equation}
for all $k = 1, \ldots, N$ and all $t$. In this way, the market portfolio is a price-weighted portfolio, since the amount it invests in each asset is proportional to the price of that asset.

The two other portfolios that are central to our theoretical and empirical analysis are the portfolios of subsets of bottom- and top-ranked assets. For any $1 \leq c < N$, the \emph{``small'' portfolio} of bottom-ranked, lower-priced assets is denoted by $s_c(t) = (s_{c, 1}(t), \ldots, s_{c, N}(t))$, and defined by
\begin{equation}
 s_{c, k}(t) = \left\{\begin{array}{ll}
        0, & \text{for }  k \leq c \\
        1/\Theta_{sc}(0), & \text{for } k > c,
        \end{array}\right.
\end{equation}
for all $k = 1, \ldots, N$.\footnote{The portfolio $s_c$ holds $1/\Theta_{sc}(0)$ shares of the assets rather than one share in order to ensure that the initial condition \eqref{valueNormalizationEq} holds. The portfolio $s_c$ and a portfolio that holds one share of the same assets have identical returns.} The portfolio $s_c$ purchases equal numbers of shares of each of the ``smallest'', lowest-priced $N - c$ assets in the market, and has weights given by
\begin{equation} \label{smallWeights}
 w^{sc}_{k}(t) = \left\{\begin{array}{ll}
        0, & \text{for }  k \leq c \\
        \frac{\theta_{(k)}(t)}{\Theta_{sc}(t)}, & \text{for } k > c,
        \end{array}\right.
\end{equation}
for all $k = 1, \ldots, N$ and all $t$. Thus, the small portfolio $s_c$ price-weights each of the $N - c$ assets it holds, in a manner similar to the market portfolio which price-weights all $N$ assets.

For any $1 \leq c \leq N$, the \emph{``big'' portfolio} of top-ranked, higher-priced assets is denoted by $b_c(t) = (b_{c, 1}(t), \ldots, b_{c, N}(t))$ and defined by
\begin{equation}
 b_{c, k}(t) = \left\{\begin{array}{ll}
        1/\Theta_{bc}(0), & \text{for }  k \leq c \\
        0, & \text{for } k > c,
        \end{array}\right.
\end{equation}
for all $k = 1, \ldots, N$. The portfolio $b_c$ purchases equal numbers of shares of each of the ``biggest'', highest-priced $c$ assets in the market, and has weights given by
\begin{equation} \label{bigWeights}
 w^{bc}_{k}(t) = \left\{\begin{array}{ll}
        \frac{\theta_{(k)}(t)}{\Theta_{bc}(t)}, & \text{for }  k \leq c \\
        0, & \text{for } k > c,
        \end{array}\right.
\end{equation}
for all $k = 1, \ldots, N$ and all $t$. Like the small portfolio $s_c$ and the market portfolio $m$, the big portfolio $b_c$ price-weights each of the $c$ assets it holds.

\subsection{Results} \label{sec:results}
In this section, we characterize the returns of the small and big portfolios relative to the market portfolio as well as relative to each other. We show that these relative returns can be decomposed into rank crossovers, as measured by a local time process, and changes in the relative price of bottom- and top-ranked asset subsets, just like in \eqref{intuitiveEq}. The following theorem, which is similar to the more general results in Theorem 3.8 and Example 3.9 of \citet{Karatzas/Ruf:2017}, establishes this fact. Its proof is in the appendix.

\begin{thm} \label{relValueThm}
For all $1 \leq c < N$, the small portfolio $s_c$ has a value process $V_{sc}$ that satisfies
\begin{equation} \label{relValueEq1}
 \log V_{sc}(T) - \log V_m(T) = \frac{1}{2} \int_0^T \frac{d\L^c(t)}{\Theta_{sc}(t)}  + \log \Theta_{sc}(T) - \log \Theta_{sc}(0),
\end{equation}
for all $T$. For all $1 \leq c \leq N$, the big portfolio $b_c$ has a value process $V_{bc}$ that satisfies
\begin{equation} \label{relValueEq2}
 \log V_{bc}(T) - \log V_m(T) =  -\frac{1}{2} \int_0^T \frac{d\L^c(t)}{\Theta_{bc}(t)} + \log \Theta_{bc}(T) - \log \Theta_{bc}(0),
\end{equation}
for all $T$.
\end{thm}

Theorem \ref{relValueThm} is powerful because it decomposes the relative value of all portfolios of bottom- and top-ranked asset subsets into cumulative rank crossovers, as measured by the non-negative local time process $d\L^c$, and the change in the relative price of the bottom- and top-ranked asset subsets. In the case of the big portfolio, these crossovers subtract from relative returns since they consist of higher-ranked, higher-priced assets dropping out of the top subset. Crucially, these portfolios are easily constructed \emph{without any knowledge of the underlying fundamentals of the assets}. The small and big portfolios $s_c$ and $b_c$ hold, respectively, equal numbers of shares of the bottom $N-c$ ranked assets and the top $c$ ranked assets at all times $t$. These asset price ranks are easily observed over time, and do not require difficult calculations or costly information acquisition.

In addition to the decompositions \eqref{relValueEq1} and \eqref{relValueEq2}, Theorem \ref{relValueThm} also yields a simple decomposition of the value of the small portfolio relative to the big portfolio. This is obtained immediately by subtracting \eqref{relValueEq2} from \eqref{relValueEq1}.

\begin{cor} \label{relValueCor}
For all $1 \leq c < N$, the value of the small portfolio $s_c$ relative to the big portfolio $b_c$ satisfies
\begin{equation} \label{relValueEq3}
\begin{aligned}
 \log V_{sc}(T) - \log V_{bc}(T) & = \frac{1}{2} \int_0^T \left( \frac{1}{\Theta_{sc}(t)} + \frac{1}{\Theta_{bc}(t)} \right) \, d\L^c(t) \\
 & \qquad \qquad + \log \big( \Theta_{sc}(T)/\Theta_{bc}(T) \big) - \log \big( \Theta_{sc}(0)/\Theta_{bc}(0) \big),
\end{aligned}
\end{equation}
for all $T$.
\end{cor}

Like Theorem \ref{relValueThm}, Corollary \ref{relValueCor} decomposes the value of the small portfolio relative to the big portfolio into cumulative rank crossovers and the change in the price of the small, bottom-ranked assets relative to the price of the big, top-ranked assets. In this decomposition, cumulative rank crossovers are measured by the local-time process $d\L^c$, as in the theorem. This decomposition is valid for any small-big rank cutoff $c$. 

The decompositions \eqref{relValueEq1} and \eqref{relValueEq2} from Theorem \ref{relValueThm} characterize the log values of the portfolios $s_c$ and $b_c$ relative to the log value of the market portfolio $m$ at time $T$ in terms of the cumulative value of the adjusted local time process, respectively, $\int_0^T \frac{d\L^c(t)}{2\Theta_{sc}(t)}$ and $-\int_0^T \frac{d\L^c(t)}{2\Theta_{bc}(t)}$, and the change in the log relative price of the bottom- and top-ranked asset subsets, respectively, $\log \Theta_{sc}(T) - \log \Theta_{sc}(0)$ and $\log \Theta_{bc}(T) - \log \Theta_{bc}(0)$. In order to go from these characterizations of relative portfolio values to a characterization of relative portfolio returns, we take differentials of both sides of \eqref{relValueEq1} and \eqref{relValueEq2}. This yields
 \begin{equation} \label{relReturnEq1}
 d \log V_{sc}(t) - d \log V_m(t) = \frac{d\L^c(t)}{2\Theta_{sc}(t)} + d\log \Theta_{sc}(t),
\end{equation}
for all $t$, and
\begin{equation} \label{relReturnEq2}
 d \log V_{bc}(t) - d \log V_m(t) = -\frac{d\L^c(t)}{2\Theta_{bc}(t)} + d\log \Theta_{bc}(t),
\end{equation}
for all $t$. According to \eqref{relReturnEq1}, the log return of the small portfolio relative to the market can be decomposed into positive rank crossovers, measured by $d\L^c \geq 0$, and changes in the relative price of the bottom-ranked assets, measured by $ d\log \Theta_{sc}$. Similarly, \eqref{relReturnEq2} states that the log return of the big portfolio relative to the market can be decomposed into negative rank crossovers, measured by $-d\L^c \leq 0$, and changes in the relative price of the top-ranked assets, measured by $d\log \Theta_{bc}$. A similar transformation of \eqref{relValueEq3} yields
\begin{equation} \label{relReturnEq3}
 d \log V_{sc}(t) - d \log V_{bc}(t) = \frac{1}{2}\left( \frac{1}{\Theta_{sc}(t)} + \frac{1}{\Theta_{bc}(t)} \right) \, d\L^c(t) + d\log \big( \Theta_{sc}(t)/\Theta_{bc}(t) \big),
\end{equation}
for all $t$. As with \eqref{relReturnEq1} and \eqref{relReturnEq2}, \eqref{relReturnEq3} states that the log return of the small portfolio relative to the big portfolio can be decomposed into positive rank crossovers, again measured by $d\L^c \geq 0$, and changes in the price of bottom-ranked assets relative to top-ranked assets, measured by $d\log ( \Theta_{sc}/\Theta_{bc} )$.

The relative return characterizations \eqref{relReturnEq1}-\eqref{relReturnEq3} are of the same form as the intuitive version \eqref{intuitiveEq} presented in the Introduction. Therefore, Theorem \ref{relValueThm} and Corollary \ref{relValueCor} imply that increases (decreases) in bottom-ranked relative asset prices raise (lower) the returns for the small portfolio relative to both the market and big portfolios. They also imply that if bottom- and top-ranked relative prices are unchanged, then the relative returns for the small portfolio will be either non-negative or positive and the relative returns for the big portfolio will be either non-positive or negative. This rank effect is a simple and necessary consequence of the non-negativity of the local time process in \eqref{relValueEq1}-\eqref{relValueEq3}. We confirm these predictions using commodity futures data in Section \ref{sec:empirics}.

Another implication of Theorem \ref{relValueThm} and Corollary \ref{relValueCor} is that one part of the decomposition of the relative values of the portfolios $s_c$ and $b_c$ is a finite variation process. In particular, the cumulative value of the adjusted local time process, $\int_0^T \frac{d\L^c(t)}{2\Theta_{sc}(t)}$, is a finite variation process by construction, as are the corresponding adjusted local time processes from \eqref{relValueEq2} and \eqref{relValueEq3}. To see why these are finite variation processes, note that the stochastic integral of a non-negative continuous process is non-decreasing and continuous, and any non-decreasing continuous process is a finite variation process \citep{Karatzas/Shreve:1991}.

Recall from Section \ref{sec:setup} that a finite variation process has finite total variation over every interval $[0, T]$. This means that the process has zero quadratic variation, or equivalently, zero instantaneous variance. In Section \ref{sec:empirics}, we decompose actual relative returns as described by Theorem \ref{relValueThm} using monthly commodity futures data and show a clear contrast between the time-series behavior of the zero-instantaneous-variance process $\int_0^T \frac{d\L^c(t)}{2\Theta_{sc}(t)}$ and the positive-instantaneous-variance process $\log \Theta_{sc}(T)$. In particular, we find that the sample variance of the finite variation process is orders of magnitude lower than that of the positive quadratic variation process, as predicted by the theorem. We find similar behavior for the finite variation processes in \eqref{relValueEq2} and \eqref{relValueEq3} as well.

The decompositions \eqref{relValueEq1}-\eqref{relReturnEq3} are little more than accounting identities, which are approximate in discrete time and exact in continuous time. There are essentially no restrictive assumptions about the underlying dynamics of asset prices and their co-movements that go into these results, making it difficult to imagine an equilibrium model of asset pricing that meaningfully deviates with Theorem \ref{relValueThm} or Corollary \ref{relValueCor}. Despite this generality, two simplifying assumptions behind these results --- that assets do not pay dividends, and that the market is closed so that there is no asset entry or exit over time --- merit further discussion.

If we were to include dividends in our framework, we would get relative value decompositions that are very similar to \eqref{relValueEq1}-\eqref{relValueEq3}. The only difference would be an extra term added to the decompositions \eqref{relValueEq1} and \eqref{relValueEq2} that measures the cumulative dividends from the small and big portfolios relative to the cumulative dividends from the market portfolio. Similarly, an extra term measuring the cumulative dividends from the small portfolio relative to the cumulative dividends from the big portfolio would need to be added to the decomposition \eqref{relValueEq3}. In the presence of dividends, then, relative capital gains could still be decomposed into rank crossovers and changes in relative prices as in Theorem \ref{relValueThm} and Corollary \ref{relValueCor}. The only complication would be an extra term that measures relative cumulative dividends as part of relative investment value.

The results of Theorem \ref{relValueThm} and Corollary \ref{relValueCor} can also be extended to include entry and exit into and out of the set of available market assets over time. This is accomplished by introducing another local time process in the decompositions \eqref{relValueEq1}-\eqref{relValueEq3} that measures the differential impact of entry and exit on the returns of the pair of portfolios being compared. We consider such an extension of our framework in Appendix \ref{sec:entry}, and provide decompositions similar to those in Theorem \ref{relValueThm} and Corollary \ref{relValueCor}.

\vskip 50pt

\section{Empirical Results} \label{sec:empirics}

Having characterized the relative returns in full generality in Section \ref{sec:theory}, we now turn to an empirical analysis. We wish to investigate the accuracy of the decompositions in Theorem \ref{relValueThm} and Corollary \ref{relValueCor} using real asset price data. In this section, we show that these decompositions provide accurate descriptions of actual relative returns for small and big portfolios of commodity futures.

\subsection{Data} \label{sec:data}
We use data on the prices of 30 different commodity futures from 1969-2018 to test our theoretical predictions. The choice to focus on commodity futures is motivated by the fact that the two most important assumptions we impose on our theoretical framework --- that assets do not pay dividends, and that the market is closed and there is no asset entry or exit over time --- align fairly closely with commodity futures markets. These assets do not pay dividends, with returns driven entirely by capital gains. Commodity futures also rarely exit from the market, which is notable since such exit can affect the returns of the small and big portfolios. In fact, no commodity futures contracts that we are aware of disappear from the market from 1969-2018, so this potential issue is irrelevant over the time period we consider. While new commodity futures contracts do enter into our data set between 1969-2018, such entry does not affect our empirical results and is easily incorporated into our framework as we explain in detail below.


Table \ref{commInfoTab} lists the start date and trading market for the 30 commodity futures in our 1969-2018 data set. These commodities encompass the four primary commodity domains (energy, metals, agriculture, and livestock). The table also reports the annualized average and standard deviation of daily log price changes over the lifetime of each futures contract. These data were obtained from the Pinnacle Data Corp., and report the two-month-ahead futures price of each commodity on each day that trading occurs, with the contracts rolled over each month.

Ranked relative asset prices, as defined by the $\theta_{(k)}$'s in \eqref{rankTheta}, are crucial to our theoretical framework and results. This concept, however, is essentially meaningless in the context of commodity prices, since different commodities are measured using different units such as barrels, bushels, and ounces. In order to give ranked prices meaning in the context of commodity futures, we normalize all contracts with data on the January 2, 1969 start date so that their prices are equal to each other. All subsequent price changes occur without modification, meaning that price dynamics are unaffected by our normalization. For those commodities that enter into our data set after 1969, we set their initial log prices equal to the average log price of those commodities already in our data set on that date. After these commodities enter into the data set with a normalized price, all subsequent price changes occur without modification. The normalized commodity futures prices we construct are similar to price indexes, with all indexes set equal to each other on the initial start date and any indexes that enter after this start date set equal to the average of the existing indexes.

Figure \ref{relPricesFig} plots the normalized log commodity futures prices relative to the average for all 30 contracts in our data set from 1969-2018. This figure shows how normalized prices quickly disperse after the initial start date, with commodity futures prices constantly being affected by different shocks. After an initial period, however, the normalized commodity futures prices are roughly stable relative to each other with what looks like only modest increases in dispersion occurring after approximately 1980. These patterns are quantified and analyzed in our empirical analysis below.

\subsection{Portfolio Construction} \label{sec:portConst}
For our empirical analysis, it is necessary to construct a market portfolio strategy that holds equal numbers of shares of each asset. In the context of commodity futures, however, the market portfolio cannot hold equal shares of each asset since futures contracts are simply agreements between two parties with no underlying asset held. This issue is easily resolved, however, since the market portfolio weights \eqref{marketWeights} are well-defined in the context of normalized commodity futures prices. In particular, \eqref{marketWeights} implies that the market portfolio invests in each commodity futures contract an amount that is proportional to the normalized price of that commodity. For this reason, we often refer to the market portfolio as the price-weighted market portfolio in the empirical analysis of this section. Note that the market portfolio of commodity futures requires no rebalancing, since price changes automatically cause the weights of each commodity in the portfolio to change in a manner that is consistent with price-weighting.

In addition to the price-weighted market portfolio, we construct small and big portfolios of low- and high-ranked commodity futures as described in Section \ref{sec:port}. The weights that define these portfolios are given by \eqref{smallWeights} and \eqref{bigWeights}, and are constructed using the normalized prices for which price rank is a meaningful concept. We set the rank cutoff between the small and big portfolios $c$ equal to half of the total number of commodity futures $N$, so that the small portfolio holds the bottom half of lower-ranked, lower-priced futures and the big portfolio holds the top half of higher-ranked, higher-priced futures.\footnote{More precisely, we set $c$ equal to the largest integer less than or equal to $N/2$.} There are two cases in which the small and big portfolios need to rebalanced. The first is if there is a change in rank so that one commodity future that was previously in the bottom half of ranked prices is now in the top half. The second is if the total number of commodity futures changes, as occurs when new futures contracts enter into our data set. The small and big portfolios are rebalanced each month if either of these two events occur, but not otherwise.


Finally, even though our commodity futures data cover 1969-2018, the fact that we normalize prices by setting them equal to each other on the 1969 start date implies that ranked relative prices will have little meaning until these prices are given time to disperse. In a manner similar to the commodity value measure of \citet{Asness/Moskowitz/Pedersen:2013}, we wait five years before forming the small, big, and price-weighted market portfolios, so that these portfolios are constructed using normalized prices from 1974-2018. Furthermore, for each commodity that enters our data set after 1969, we wait five years after the entrance date before including that new commodity in the small, big, and price-weighted market portfolios. Because the initial log prices of newly entering commodity futures are set equal to the average log price of the commodities already in our data set, this five-year wait period allows the price of newly entering commodities to disperse in a manner that gives ranked relative prices more meaning.

\subsection{Results}
Figure \ref{returnsFig} plots the log cumulative returns for the price-weighted market portfolio and the small and big portfolios of bottom- and top-ranked commodity futures from 1974-2018. The figure shows that all three portfolios have roughly similar behavior over time, but that the small portfolio consistently outperforms both the price-weighted and big portfolios over time, with the big portfolio performing the worst of all. These patterns are quantified in Table \ref{returnsTab}, which reports the annualized average and standard deviation of monthly returns for all three portfolios over this time period. The monthly returns of the market portfolio have correlations of 0.76 and 0.96 with the returns of the small and big portfolios, respectively, while the returns of the small and big portfolios have a correlation of 0.55. The outperformance of the small portfolio relative to both the price-weighted market and big portfolios is also evident in Table \ref{relReturnsTab}, which reports the annualized average, standard deviation, and Sharpe ratio of monthly relative returns for the small, bottom-ranked portfolio from 1974-2018. Tables \ref{returnsTab} and \ref{relReturnsTab} also report returns statistics for each decade in our long sample period.

The results of Tables \ref{returnsTab} and \ref{relReturnsTab} show that the small portfolio consistently and substantially outperformed the price-weighted market portfolio from 1974-2018, and that the big portfolio consistently and substantially underperformed the price-weighted market portfolio over this time period. The rank effect is most evident from the high Sharpe ratios for the excess returns of the small portfolio relative to the price-weighted and big portfolios, as shown in Table \ref{relReturnsTab}. Notably, both of these Sharpe ratios consistently rise to 0.6 or higher after 1980, which is after most of the commodity futures contracts in our data set have started trading, as shown by Table \ref{commInfoTab}. Thus, as the number of tradable assets $N$ rises, portfolio outperformance also rises. This is not surprising, since a greater number of tradable assets generally implies more rank crossovers and hence a greater value for the non-negative local time process, $d\L^c$, and it is this process that mostly determines relative portfolio returns over long horizons, as we demonstrate below.

The general theory of Section \ref{sec:theory} does not make any statements about the size of portfolio returns. Instead, this theory states that the returns for the small and big portfolios relative to the market can be decomposed into rank crossovers and changes in the relative price of bottom- and top-ranked asset subsets, according to Theorem \ref{relValueThm}. The theory also states that the returns for the small portfolio relative to the big portfolio can be decomposed into rank crossovers and changes in the price of the bottom-ranked asset subset relative to the top-ranked asset subset, according to Corollary \ref{relValueCor}

In order to empirically investigate the decomposition \eqref{relValueEq1} from Theorem \ref{relValueThm}, in Figure \ref{returnsSmallFig} we plot the cumulative abnormal returns --- returns relative to the price-weighted market portfolio --- of the small portfolio together with the cumulative value of rank crossovers, $\int_0^T \frac{d\L^c(t)}{2\Theta_{sc}(t)}$, from 1974-2018. In addition, Figure \ref{relPSmallFig} plots the relative price of the small, bottom-ranked commodity futures, $\Theta_{sc}$, normalized relative to its average value for 1974-2018. In addition to the consistent and substantial outperformance of the small portfolio relative to the price-weighted portfolio, these figures show that short-run relative return fluctuations for the small portfolio closely follow fluctuations in the relative price of the bottom-ranked commodities while the long-run behavior of these relative returns closely follow the smooth crossovers. Indeed, there is a striking contrast between the high volatility of relative prices in Figure \ref{relPSmallFig} and the low volatility of crossovers in Figure \ref{returnsSmallFig}. This is an important observation that is a direct prediction of Theorem \ref{relValueThm}, a point we discuss further below.

In addition to the contrasting volatilities of relative prices and crossovers, Figures \ref{returnsSmallFig} and \ref{relPSmallFig} show that the cumulative abnormal returns of the small portfolio are equal to the cumulative value of the adjusted local time process, $\int_0^T \frac{d\L^c(t)}{2\Theta_{sc}(t)}$, plus the change in the log of the relative price of bottom-ranked commodities, $\log \Theta_{sc}(T) - \log \Theta_{sc}(0)$. Indeed, the solid black line in Figure \ref{returnsSmallFig} (cumulative abnormal returns) is equal to the dashed red line in that same figure (cumulative value of the adjusted local time process) minus the line in Figure \ref{relPSmallFig} (change in the log of the relative price of bottom-ranked commodities). This is exactly the relationship described by \eqref{relValueEq1} from Theorem \ref{relValueThm}. We stress, however, that this empirical relationship is a necessary consequence of how the non-decreasing local time process, $\L^c$, is calculated. For each day that we have data, the value of $\L^c$ for that day is calculated by subtracting the log value of the relative price of bottom-ranked commodities, $\log \Theta_{sc}$, from the cumulative abnormal returns, $\log V_{sc} - \log V_m$, and then adjusting according to the identity \eqref{relValueEq1} from Theorem \ref{relValueThm}.

Given that the empirical decomposition of Figures \ref{returnsSmallFig} and \ref{relPSmallFig} is constructed so that \eqref{relValueEq1} must hold, it is natural to wonder what the usefulness of this decomposition is. One useful aspect of the decomposition \eqref{relValueEq1} lies in the prediction that one part of this decomposition, the cumulative value of the adjusted local time process that measures rank crossovers, is non-decreasing. This prediction is  confirmed by the relatively smooth upward slope of the rank crossovers line in Figure \ref{returnsSmallFig}, and has implications for the long-run relative performance of the small and price-weighted portfolios, as we discuss below. 

Another key prediction of the decomposition \eqref{relValueEq1} is that the cumulative value of the adjusted local time process is a finite variation process, while the other part, the change in the log of the relative price of bottom-ranked commodities, $\log \Theta_{sc}$, is not. Recall from the discussion in Section \ref{sec:setup} that a finite variation process has zero quadratic variation, or zero instantaneous variance. To be clear, the prediction that the cumulative value of the adjusted local time is a finite variation process is not a prediction that the sample variance of changes in the cumulative value of the adjusted local time computed using monthly, discrete-time data will be equal to zero, but rather a prediction that these changes will be roughly constant over time.\footnote{Note that the sample variance of a continuous-time finite variation process computed using discrete-time data will never be exactly equal to zero.} Our results thus predict that the cumulative value of the adjusted local time process will grow at a roughly constant rate with only few and small changes over time.

This smooth growth is exactly what is observed in the dashed red line of Figure \ref{returnsSmallFig}, and, as mentioned above, is in stark contrast to the highly volatile behavior of bottom-ranked relative prices shown in Figure \ref{relPSmallFig}. This contrast can be quantified by noting that the coefficient of variation of changes in the cumulative value of the adjusted local time is equal to 7.57, while the coefficient of variation of changes in the relative price of bottom-ranked commodities is equal to 326.06. These results confirm one of the key predictions of Theorem~\ref{relValueThm}.


The positive and relatively constant rank crossovers over time, measured by the local time $d\L^c$, have an important implication for the long-run return of the small portfolio relative to the price-weighted market portfolio. Since Theorem \ref{relValueThm} and \eqref{relReturnEq1} imply that relative returns can be decomposed into rank crossovers and changes in relative prices, consistently positive crossovers over long time horizons can only be counterbalanced by a consistently falling relative price for bottom-ranked assets. In the absence of such falling relative prices, the positive crossovers guarantee outperformance relative to the market. Therefore, the small decline in the relative price of the small, bottom-ranked commodity futures shown in Figure \ref{relPSmallFig} together with the positive crossovers shown in Figure \ref{returnsSmallFig} ensure the existence of a rank effect in which the small portfolio outperforms the market portfolio over the 1974-2018 time period.

In a similar manner to Figure \ref{returnsSmallFig}, Figure \ref{returnsBigFig} plots the cumulative returns of the big portfolio relative to the price-weighted market portfolio together with the cumulative value of the adjusted local time process $-\frac{d\L^c}{2\Theta_{bc}}$ over the 1974-2018 time period. Also, Figure \ref{relPBigFig} plots the log relative price of top-ranked commodities over this same time period. Similarly, Figure \ref{returnsSmallBigFig} plots the cumulative returns of the small portfolio relative to the big portfolio together with the cumulative value of the adjusted local time process $\frac{1}{2} \left( \frac{1}{\Theta_{sc}} + \frac{1}{\Theta_{bc}} \right) \, d\L^c$, while Figure \ref{relPSmallBigFig} plots the log prices of bottom-ranked commodities relative to top-ranked commodities. The cumulative values of the adjusted local time processes in Figures \ref{returnsBigFig} and \ref{returnsSmallBigFig} are calculated using the identities \eqref{relValueEq2} and \eqref{relValueEq3} from Theorem \ref{relValueThm} and Corollary \ref{relValueCor}. The results in Figures \ref{returnsBigFig}-\ref{relPSmallBigFig} for the big and small portfolios align closely with the results in Figures \ref{returnsSmallFig} and \ref{relPSmallFig}.


Figures \ref{returnsBigFig}-\ref{relPSmallBigFig} show that short-run return fluctuations for both the big portfolio relative to the market and the small portfolio relative to the big portfolio closely follow fluctuations in the relative prices of the corresponding ranked asset subsets. In contrast, the long-run behavior of these relative returns follow the smoothly accumulating rank crossovers. Much like in Figure \ref{returnsSmallFig}, Figures \ref{returnsBigFig} and \ref{returnsSmallBigFig} show that the cumulative values of rank crossovers grow at roughly constant rates over time (negative in the case of the big portfolio), with a clear contrast between this stable growth and the large fluctuations in relative prices shown in Figures \ref{relPBigFig} and \ref{relPSmallBigFig}. As discussed above, the fact that adjusted rank crossovers, which depend on the local time process $\L^c$, are approximately constant over time is consistent with the prediction that their cumulative value is a finite variation process, thus confirming one of the key results in Theorem \ref{relValueThm} and Corollary \ref{relValueCor}. Finally, Figures \ref{returnsBigFig} and \ref{returnsSmallBigFig} confirm the consistent and substantial outperformance of, respectively, the price-weighted market portfolio relative to the big portfolio and the small portfolio relative to the big portfolio. As with the small portfolio relative to the market, this long-run rank effect is predicted by \eqref{relValueEq1}-\eqref{relValueEq3} and Theorem \ref{relValueThm} and Corollary \ref{relValueCor} given the small change in relative prices observed in Figures \ref{relPBigFig} and \ref{relPSmallBigFig} as compared to the large changes in the cumulative value of rank crossovers observed in Figures \ref{returnsBigFig} and \ref{returnsSmallBigFig}.

In order to investigate the robustness of our results to the price normalization start date of January 2, 1969, in Figure \ref{varyStartFig} we report the Sharpe ratios of monthly returns for the small, bottom-ranked portfolio relative to the price-weighted market and big, top-ranked portfolios for a range of data start dates from 1969-2000. More specifically, we create normalized commodity futures prices as described in Section \ref{sec:data} (and shown in Figure \ref{relPricesFig}) with different normalization start dates that range across every quarter from 1969 to the start of 2000. We then construct the small, big, and market portfolios as described in Section \ref{sec:portConst} using each of these different normalized commodity futures price data sets, and then report relative return Sharpe ratios for each data set in Figure \ref{varyStartFig}. These different data sets are identified by their different price normalization start dates. The figure clearly shows that our results are not sensitive to the price normalization start date of January 2, 1969, which is used to generate our main results in Tables \ref{returnsTab}-\ref{confacTab} and Figures \ref{returnsFig}-\ref{relPSmallBigFig} and was chosen only because it maximizes the length of our data sample. In fact, Figure \ref{varyStartFig} shows that most price normalization start dates after 1969 generate even larger Sharpe ratios, with the average across all dates shown in the figure equal to 0.66 for both the small portfolio relative to the market portfolio and the small portfolio relative to the big portfolio --- noticeably higher than the full-sample Sharpe ratios reported in Table \ref{relReturnsTab}.

\vskip 50pt

\section{Discussion} \label{sec:discussion}

The empirical results shown in Figures \ref{returnsSmallFig}-\ref{relPSmallBigFig} confirm the prediction of Theorem \ref{relValueThm} and Corollary \ref{relValueCor} that the rank crossover components of the decompositions \eqref{relValueEq1}-\eqref{relValueEq3} are nearly constant. Furthermore, these figures show that cumulative rank crossovers are positive and increasing in the case of the returns of the small portfolio relative to both the price-weighted market and big portfolios. In contrast, cumulative rank crossovers are negative and decreasing in the case of the returns of the big portfolio relative to the market, as shown in Figure~\ref{returnsBigFig}.

\subsection{Relative Prices and Relative Returns}
Taken together, our theoretical and empirical results confirm that changes in the relative prices of ranked asset subsets are key determinants of the returns for portfolios of ranked assets relative to the market. In this sense, the relative prices of different ranked asset subsets are an asset pricing factor, since their fluctuations drive relative returns for all price-weighted portfolios of ranked assets. Crucially, the theoretical results of Theorem \ref{relValueThm} and Corollary \ref{relValueCor} that establish this relative price asset pricing factor are achieved under minimal assumptions that should be consistent with essentially any model of asset pricing, meaning that this factor is universal across different economic and financial environments. Our empirical results in Figures \ref{returnsSmallFig}-\ref{relPSmallBigFig} help to confirm this universality.

Our model-free, mathematical approach is similar to the approach of \citet{Fernholz/Stroup:2018}, and the relative price asset pricing factor we describe is related to the price dispersion asset pricing factor analyzed by these authors. As with the price dispersion factor, then, the relative price factor is not subject to the criticisms that have been raised recently about the implausibly high and rising number of factors and anomalies that the empirical asset pricing literature has identified. \citet{Novy-Marx:2014}, \citet{Harvey/Liu/Zhu:2016}, and \citet{Bryzgalova:2016}, for example, have pointed to different econometric issues with many of these proposed factors. In contrast, the relative price asset pricing factor established by Theorem \ref{relValueThm} and Corollary \ref{relValueCor} is not derived using a specific economic model or a specific regression framework, but rather using general mathematical methods that represent asset prices as continuous semimartingales that are consistent with essentially all models and empirical specifications. In this way, our results are naturally immune to these criticisms.

\subsection{Relative Prices with no Rank Effect}
What would commodity futures prices look like in 2018 if the small portfolio had not outperformed the price-weighted market and big portfolios? The decompositions \eqref{relValueEq1}-\eqref{relValueEq3} allow us to answer this difficult question. In particular, these decompositions imply that there is no rank effect only if the relative price of bottom-ranked, lower-priced commodities consistently and rapidly declines. Although the relative price of the bottom ranks did decline slightly from 1974-2018, as can be seen in Figures \ref{relPSmallFig} and \ref{relPSmallBigFig}, this decline was far too small to negate the large gains from rank crossovers, as seen in Figures \ref{returnsSmallFig} and \ref{returnsSmallBigFig}. Using \eqref{relValueEq1}-\eqref{relValueEq3}, we can quantify just how much larger a decline in the relative price of bottom-ranked commodities was needed to eliminate the rank effect from 1974-2018, and then examine the implications of such a decline for commodity futures prices in 2018.

Table \ref{confacTab} reports the counterfactual relative price of bottom-ranked, lower-priced commodity futures needed in 2018 for the small portfolio to not have outperformed the price-weighted market and big portfolios from 1974-2018. In the case of the small portfolio relative to the market, the counterfactual number is obtained by altering the relative price $\Theta_{sc}(T)$ from \eqref{relValueEq1} so that the relative value $\log V_{sc}(T) - \log V_m(T)$ is equal to zero. Implicitly, then, this calculation assumes that the cumulative value of rank crossovers from 1974-2018, measured by $\int_0^T \frac{d\L^c(t)}{2\Theta_{sc}(t)}$, is unchanged from its true value over this time period. A similar assumption aids this same calculation for the small portfolio relative to the big portfolio. In this case, the counterfactual number is obtained by altering the relative price $\Theta_{sc}(T)/\Theta_{bc}(T)$ from \eqref{relValueEq3} so that the relative value $\log V_{sc}(T) - \log V_{bc}(T)$ is equal to zero. 

According to Table \ref{confacTab}, the counterfactual relative price of bottom-ranked, lower-priced commodity futures needed in 2018 to eliminate the rank effect is more than an order of magnitude lower than the actual relative price in 2018. This difference is particularly pronounced for the change in the price of bottom-ranked commodities relative to top-ranked commodities, which corresponds to the return of the small portfolio relative to the big portfolio. This is no surprise given that the small portfolio generates 9.82\% higher returns on average than the big portfolio from 1974-2018, as reported in Table \ref{relReturnsTab}. In both cases, however, the counterfactual relative prices of bottom-ranked commodities are so low that they imply commodity futures prices in 2018 that are radically different from those actually observed in 2018. 

Table \ref{confacTab} plainly shows these significant differences. According to the table, the price of the 15 top-ranked commodity futures in 2018 would have had to be approximately 274 times larger than the price of the 15 bottom-ranked commodity futures in 2018 in order for there to be no small-versus-big rank effect from 1974-2018. This ratio is more than 100 times larger than the actual ratio observed at the end of our data set in 2018. This counterfactual ratio translates into either a much lower counterfactual price of bottom-ranked commodities, a much higher counterfactual price of top-ranked commodities, or some intermediate mix of the two. More specifically, if the price of the 15 bottom-raked commodities in 2018 had been more than 100 times lower than actually observed in 2018, there would have been no rank effect. However, such prices would have implied historically unprecedented lows for any of these commodities. Alternatively, if the price of the 15 top-ranked commodities in 2018 had been more than 100 times higher than actually observed in 2018, there also would have been no rank effect. In this case, the price of any of these commodities would have reached historically unprecedented highs in 2018. In both cases, the counterfactuals in which the rank effects from 1974-2018 shown in Figures \ref{returnsBigFig} and \ref{returnsSmallBigFig} disappear imply radically different and historically unprecedented commodity prices for 2018.


\subsection{The Rank Effect and the Value Anomaly}
The rank effect relative returns reported in Table \ref{relReturnsTab} are similar to the value anomaly for commodity futures described by \citet{Asness/Moskowitz/Pedersen:2013}. These authors rank commodities based on their current price relative to their average price 4.5-5.5 years in the past, and show that portfolios consisting of high-value commodity futures --- those contracts with low prices today relative to the past --- consistently and substantially outperform portfolios consisting of low-value commodity futures --- those contracts with high prices today relative to the past. This ranking of commodity futures based on their current price relative to their past price is similar to our normalized commodity price rankings in which the normalized prices are equalized on the 1969 data start date, as discussed in Section \ref{sec:data}. 

In order to quantify the extent of this similarity, in Table \ref{valueTab} we report the results of time series regressions of value anomaly excess returns similar to those uncovered by \citet{Asness/Moskowitz/Pedersen:2013} on the returns for the small portfolio relative to the market and the small portfolio relative to the big portfolio. For the value anomaly, each month we rank commodity futures prices based on their current price relative to their average price 4.5-5.5 years in the past. We then examine the return of a portfolio that price-weights each of the bottom-ranked half of commodity futures relative to a portfolio that price-weights each of the top-ranked commodity futures, with both portfolios rebalanced each month. According to the results in Table \ref{valueTab}, both versions of the rank effect are similar to the value anomaly for commodities.


The results in Table \ref{valueTab} imply that much of the value anomaly for commodity futures uncovered by \citet{Asness/Moskowitz/Pedersen:2013} is a direct consequence of the approximate stability of the price of bottom-ranked commodities relative to top-ranked commodities from 1974-2018, as shown in Figure \ref{relPSmallBigFig}. Furthermore, Table \ref{confacTab} implies that for this value anomaly to not have existed, commodity futures prices in 2018 would have had to be radically different than what they were. Indeed, the relative price of bottom-ranked commodities would have had to be more than an order of magnitude lower than they were. This result offers a novel interpretation of the value anomaly for commodity futures, and implies that any explanation of this anomaly must also explain the approximate stability of ranked relative commodity prices. The value anomaly would not have existed only if there had been a dramatic and unprecedented split in the price of bottom- and top-ranked commodity futures prices from 1974-2018.

\subsection{The Rank Effect and Efficient Markets}
The relative return decompositions of Theorem \ref{relValueThm} and Corollary \ref{relValueCor} reveal a novel dichotomy for markets in which dividends and asset entry/exit over time play small roles. On the one hand, the relative prices of bottom-ranked, lower-priced assets may be stable over time, in which case \eqref{intuitiveEq} implies that predictable excess returns exist for portfolios of those lower-ranked assets. This is the scenario we observe for commodity futures in Figures \ref{returnsSmallFig}-\ref{relPSmallBigFig}. In such markets, fluctuations in the relative prices of ranked asset subsets are linked to excess returns via the accounting identities \eqref{relReturnEq1}-\eqref{relReturnEq3}. In a standard equilibrium model of asset pricing, these predictable excess returns may exist only if they are compensation for risk. This risk, in turn, is defined by an endogenous stochastic discount factor that is linked to the marginal utility of economic agents. It is not clear, however, how marginal utility might be linked to the relative prices of bottom-ranked, lower-priced assets. It is also not clear why marginal utility should be higher when the relative price of lower-ranked assets falls, yet these are necessary implications of any standard asset pricing model in which relative prices are asymptotically stable, according to our results.

On the other hand, the relative prices of bottom-ranked, lower-priced assets may not be stable over time. In this case, the relative price of lower-ranked assets is consistently and rapidly falling, and the decomposition \eqref{intuitiveEq} no longer predicts excess returns. Instead, this decomposition predicts falling relative prices that cancel out the non-negative rank crossovers in \eqref{intuitiveEq} on average over time. The relative return decompositions of Theorem \ref{relValueThm} and Corollary \ref{relValueCor} make no predictions about the stability of the relative prices of ranked asset subsets, so this possibility is not ruled out by our theoretical results. Nonetheless, it is notable that our empirical results for commodity futures are inconsistent with this no-stability, no-excess-returns market structure. In light of these results, future work that examines the long-run properties of ranked relative prices in different asset markets and attempts to distinguish between the two sides of this dichotomy --- asymptotically stable markets with predictable excess returns versus asymptotically unstable markets without predictable excess returns --- may yield interesting new insights.

This novel dichotomy has several implications. First, it provides a new interpretation of market efficiency in terms of a constraint on cross-sectional asset price dynamics and the relative price of ranked asset subsets. Either the relative price of bottom-ranked, lower-priced assets falls consistently and rapidly over time, consistent with this constraint, or there exists a market inefficiency or a risk factor based on the decompositions \eqref{relReturnEq1}-\eqref{relReturnEq3}. Second, it raises the possibility that the well-known size effect for equities \citep{Banz:1981,Fama/French:1993} may be interpretable in terms of the dynamics of the relative size of bottom-ranked, smaller companies. To the extent that the decompositions of Theorem \eqref{relValueThm} and Corollary \ref{relValueCor} are universal, the predictable excess returns underlying the size anomaly should be linkable to a violation of the constraint on cross-sectional asset price dynamics and the relative value of smaller companies as discussed above. 

Under this interpretation, the size effect for equities may be a simple consequence of the decomposition \eqref{intuitiveEq} and the stability of the relative size of lower-ranked, smaller companies. This interpretation is supported by the empirical results of both \citet{Fernholz:1998} and \citet{Fama/French:2007}. These authors show that small companies growing into large companies --- much like the rank crossovers of \eqref{intuitiveEq} --- explains the majority of the size effect. Future work that more closely examines the size effect using the decompositions of Theorem \ref{relValueThm} and Corollary \ref{relValueCor} may yield interesting new insights about the relationship between the size distribution of company market capitalizations and the size anomaly.

\vskip 50pt

\section{Conclusion} \label{sec:conclusion}

We represent asset prices as general continuous semimartingales and show that the returns for portfolios of ranked assets can be decomposed into rank crossovers and changes in the relative price of ranked asset subsets. Because of the minimal assumptions underlying this result, our decomposition is little more than an accounting identity that is consistent with essentially any asset pricing model. We show that rank-crossovers are approximately constant over time and positive for portfolios of bottom-ranked assets. This conclusion implies that in a closed, dividend-free market in which the relative price of bottom-ranked assets is approximately constant, a portfolio of bottom-ranked assets must necessarily outperform the market portfolio over time --- a rank effect emerges. We confirm our theoretical predictions using commodity futures, and show that a portfolio of bottom-ranked, lower-priced commodities consistently and substantially outperformed the price-weighted market portfolio from 1974-2018. We show that if this rank effect did not exist, then bottom-ranked prices would have had to be much lower relative to top-ranked prices in 2018, implying radically different commodity prices than those actually observed.

\pagebreak

\begin{spacing}{1.2}

\bibliographystyle{chicago}
\bibliography{econ}

\end{spacing}

\vskip 50pt

\begin{spacing}{1.1}

\appendix
\section{Proofs} \label{sec:proofs}

This appendix presents the proof of Theorem \ref{relValueThm}. Note that Corollary \ref{relValueCor} is an immediate consequence of this theorem, since \eqref{relValueEq3} is obtained by subtracting \eqref{relValueEq2} from \eqref{relValueEq1}.

\begin{proofThm}
Theorem \ref{relValueThm} follows from the more general results in Theorem 3.8 and Example 3.9 of \citet{Karatzas/Ruf:2017}. For any $1 \leq c \leq N$, let $F_c : \R^N \to \R$ be the function
\begin{equation*}
 F_c(x_1, \ldots, x_N) = \sum_{i=1}^c x_i.
\end{equation*}
The function $F_c$ is regular for the ranked asset share processes $\theta_{(1)}, \ldots, \theta_{(N)}$ according to Definition 3.1 and Theorem 3.8 of \citet{Karatzas/Ruf:2017}, since it is continuous and concave and we have assumed that prices are always positive. Because we have assumed that there are no triple points, it follows that the local time processes $\L_{\theta_{(k)} - \theta_{(k + \ell)}} = 0$, for all $\ell \geq 2$ and all $k = 1, \ldots, N-1$. Therefore, Example 3.9 and Proposition 4.8 of \citet{Karatzas/Ruf:2017} imply that
\begin{equation*}
\begin{aligned}
 V_{bc}(T)/V_m(T) & = \frac{F_c(\theta_{(1)}(T), \ldots, \theta_{(N)}(T))}{F_c(\theta_{(1)}(0), \ldots, \theta_{(N)}(0))} \exp \left( \int_0^T \frac{d\Gamma^{F_c}(t)}{F_c(\theta_{(1)}(t), \ldots, \theta_{(N)}(t))} \right) \\
 & = \frac{\Theta_{bc}(T)}{\Theta_{bc}(0)} \exp \left( \int_0^T \frac{d\Gamma^{F_c}(t)}{\Theta_{bc}(t)} \right),
\end{aligned}
\end{equation*}
for all $T$, with
\begin{equation*}
 \Gamma^{F_c}(t) = -\frac{1}{2}\L^c(t),
\end{equation*}
for all $t$. The result \eqref{relValueEq2} follows. A similar argument using the function
\begin{equation*}
 F_c(x_1, \ldots, x_N) = \sum_{i=c+1}^N x_i,
\end{equation*}
establishes \eqref{relValueEq1} for the small portfolio $s_c$.
\end{proofThm}

\section{Asset Entry and Exit} \label{sec:entry}

This appendix considers the impact of asset entry and exit on the decompositions in Theorem \ref{relValueThm} and Corollary \ref{relValueCor}. In the main analysis, we assumed that the number of assets in the market was fixed at $N$ and that these assets maintained positive prices at all times $t$. We wish to relax that assumption.

Consider a market that consists of $N > 1$ assets, but assume that only the top $n < N$ ranked assets are tradable at each time $t$. As discussed in Section \ref{sec:rank}, assets are ranked by their prices, so the top $n$ assets that are tradable at each time $t$ are the top $n$ most expensive assets at that time $t$. An asset \emph{enters} the tradable market when it crosses over and joins the top $n$ ranks, and it \emph{exits} the tradable market when it crosses over and drops out of the top $n$ ranks. 

Because only the top $n$ ranked assets are tradable at each time $t$, it is useful to modify the definition of ranked relative prices $\theta_{(k)}$ from Section \ref{sec:rank} in terms of tradable assets only. For all $k = 1, \ldots, n$, let
\begin{equation} \label{rankThetaApp}
 \theta'_{(k)}(t) = \frac{p_{(k)}(t)}{\sum_{\ell=1}^n p_{(\ell)}(t)}.
\end{equation}
We modify the definitions of the relative prices of bottom- and top-ranked asset subsets $\Theta_{sc}$ and $\Theta_{bc}$ similarly, so that for any cutoff $1 \leq c < n$,
\begin{equation*}
 \Theta'_{sc}(t) = \sum_{k=c+1}^n \theta'_{(k)}(t),
\end{equation*}
and, for any cutoff $1 \leq c \leq n$,
\begin{equation*}
 \Theta'_{bc}(t) = \sum_{k=1}^c \theta'_{(k)}(t).
\end{equation*}
The big portfolio $b_c$ maintains the same definition as in Section \ref{sec:port}, with the extra restriction that the rank cutoff $c$ is weakly less than the number of tradable assets $n$. This restriction ensures that the big portfolio only holds assets that are tradable at each time $t$, and that it immediately sells any assets that become non-tradable by dropping out of the top $n$ ranks. 

The small portfolio $s_c$ and market portfolio $m$ also require modified definitions in this setup with asset entry and exit. Let the modified market portfolio $m'(t) = (m'_{1}(t), \ldots, m'_{N}(t))$ be defined by
\begin{equation*}
 m'_{k}(t) = \left\{\begin{array}{ll}
        1/\Theta_{bc}(0), & \text{for }  k \leq n \\
        0, & \text{for } k > n, \\
        \end{array}\right.
\end{equation*}
for all $k = 1, \ldots, N$.\footnote{The portfolio $m'$ holds $1/\Theta_{bc}(0)$ shares of the assets rather than one share in order to ensure that the initial condition \eqref{valueNormalizationEq} holds. The portfolio $m'$ and a portfolio that holds one share of the same assets have identical returns.} The portfolio $m'$ holds equal shares of each of the $n$ tradable assets in the economy at each time $t$. Because these $n$ tradable assets are also the $n$ top-ranked assets, the modified market portfolio $m'$ is equivalent to the big portfolio $b_n$ defined in Section \ref{sec:port}. The portfolio $m'$ has weights equal to the modified ranked relative prices $\theta'_{(k)}$ defined in \eqref{rankThetaApp}, since
\begin{equation} \label{marketWeightsApp}
 w^{m'}_{k}(t) = \frac{p_{(k)}(t)}{\Theta_{bc}(0)V_{m'}(t)} = \frac{\Theta_{bc}(0)p_{(k)}(t)}{\Theta_{bc}(0)\sum_{\ell=1}^n p_{(\ell)}(t)} = \theta'_{(k)}(t),
\end{equation}
for all $k = 1, \ldots, n$ and all $t$. 

For any rank cutoff $1 \leq c < n$, let the modified small portfolio $s'_c(t) = (s'_{c, 1}(t), \ldots, s'_{c, N}(t))$ be defined by
\begin{equation*}
 s'_{c, k}(t) = \left\{\begin{array}{ll}
        0, & \text{for }  k \leq c \\
        \frac{1}{\Theta_{sc}(0) - \Theta_{sn}(0)}, & \text{for } n \geq k > c \\
        0, & \text{for } k > n,
        \end{array}\right.
\end{equation*}
for all $k = 1, \ldots, N$.\footnote{As in Section \ref{sec:port}, the portfolio $s'_c$ holds $\frac{1}{\Theta_{sc}(0) - \Theta_{sn}(0)}$ shares of the assets rather than one share in order to ensure that the initial condition \eqref{valueNormalizationEq} holds. The portfolio $s'_c$ and a portfolio that holds one share of the same assets have identical returns.} The portfolio $s'_c$ purchases equal numbers of shares of each of the $n - c$ smallest, lowest-priced tradable assets in the market. Notably, this portfolio does not hold any shares of the non-tradable assets below rank $n$, a modification that distinguishes $s'_c$ from the standard small portfolio $s_c$ of Section \ref{sec:port}. The modified small portfolio $s'_c$ has weights given by
\begin{equation*}
 w^{s'c}_{k}(t) = \left\{\begin{array}{ll}
        0, & \text{for }  k \leq c \\
        \frac{\theta'_{(k)}(t)}{\Theta'_{sc}(t)}, & \text{for } n \geq k > c \\
        0, & \text{for } k > n
        \end{array}\right.
\end{equation*}
for all $k = 1, \ldots, N$ and all $t$. Thus, the portfolio $s_c$ price-weights each of the $n - c$ tradable assets it holds.

In this setup that includes asset entry and exit, we have the following extension of Theorem \ref{relValueThm}.

\begin{thm} \label{relValueThmApp}
For all $1 \leq c < n$, the modified small portfolio $s'_c$ has a value process $V_{s'c}$ that satisfies
\begin{equation} \label{relValueEq1App}
 \log V_{s'c}(T) - \log V_{m'}(T) = \frac{1}{2} \int_0^T \left( \frac{d\L^n(t)}{\Theta_{bn}(t)} + \frac{d\L^c(t) - d\L^n(t)}{\Theta_{bn}(t) - \Theta_{bc}(t)} \right) + \log \Theta'_{sc}(T) - \log \Theta'_{sc}(0),
\end{equation}
for all $T$. For all $1 \leq c \leq n$, the big portfolio $b_c$ has a value process $V_{bc}$ that satisfies
\begin{equation} \label{relValueEq2App}
 \log V_{bc}(T) - \log V_{m'}(T) =  \frac{1}{2} \int_0^T \left( \frac{d\L^n(t)}{\Theta_{bn}(t)} - \frac{d\L^c(t)}{\Theta_{bc}(t)} \right) + \log \Theta'_{bc}(T) - \log \Theta'_{bc}(0),
\end{equation}
for all $T$.
\end{thm}

\begin{proof}
First, note that \eqref{relValueEq2App} is an immediate consequence of \eqref{relValueEq2} from Theorem \ref{relValueThm}. This follows because the modified market portfolio $m'$ is equivalent to the big portfolio $b_n$, and hence we can use \eqref{relValueEq2} to characterize relative returns for both $b_c$ and $m'$ and then take the difference to get \eqref{relValueEq2App}.

For \eqref{relValueEq1App}, let $F_{cn} : \R^N \to \R$ be the function 
\begin{equation*}
 F_{cn}(x_1, \ldots, x_N) = \sum_{i=c+1}^n x_i,
\end{equation*}
where $1 \leq c < n$. The function $F_{cn}$ is regular for the ranked asset share processes $\theta_{(1)}, \ldots, \theta_{(N)}$ according to Definition 3.1 and Theorem 3.8 of \citet{Karatzas/Ruf:2017}, since it is continuous and concave and we have assumed that prices are always positive. Because we have assumed that there are no triple points, it follows that the local time processes $\L_{\theta_{(k)} - \theta_{(k+\ell)}} = 0$, for all $\ell \geq 2$ and all $k = 1, \ldots, N-1$. Therefore, Example 3.9 and Proposition 4.8 of \citet{Karatzas/Ruf:2017} imply that
\begin{align}
 V_{s'c}(T)/V_m(T) & = \frac{F_{cn}(\theta_{(1)}(T), \ldots, \theta_{(N)}(T))}{F_{cn}(\theta_{(1)}(0), \ldots, \theta_{(N)}(0))} \exp \left( \int_0^T \frac{d\Gamma^{F_{cn}}(t)}{F_{cn}(\theta_{(1)}(t), \ldots, \theta_{(N)}(t))} \right) \notag \\
 & = \frac{\Theta_{bn}(T) - \Theta_{bc}(T)}{\Theta_{bn}(0) - \Theta_{bc}(0)} \exp \left( \int_0^T \frac{d\Gamma^{F_{cn}}(t)}{\Theta_{bn}(t) - \Theta_{bc}(t)} \right), \label{relValueEqProof1App}
\end{align}
for all $T$, with
\begin{equation*}
 \Gamma^{F_{cn}}(t) = \frac{\L^c(t) - \L^n(t)}{2},
\end{equation*}
for all $t$. Finally, note that Theorem \ref{relValueThm} applied to the modified market portfolio $m'$ implies that
\begin{equation} \label{relValueEqProof2App}
\log V_{m'}(T) - \log V_m(T) = -\frac{1}{2} \int_0^T \frac{d\L^n(t)}{\Theta_{bn}(t)} + \log \Theta_{bn}(T) - \log \Theta_{bn}(0),
\end{equation}
for all $T$. The result \eqref{relValueEq1App} follows from subtracting \eqref{relValueEqProof2App} from the log of \eqref{relValueEqProof1App}.
\end{proof}

We can extend Corollary \ref{relValueCor} to this setup that includes asset entry and exit by subtracting \eqref{relValueEq2App} from \eqref{relValueEq1App}.

\begin{cor} \label{relValueCorApp}
For all $1 \leq c < n$, the value of the modified small portfolio $s'_c$ relative to the big portfolio $b_c$ satisfies
\begin{equation} \label{relValueEq3App}
\begin{aligned}
 \log V_{s'c}(T) - \log V_{bc}(T) & = \frac{1}{2} \int_0^T \left( \frac{d\L^c(t)}{\Theta_{bc}(t)} + \frac{d\L^c(t) - d\L^n(t)}{\Theta_{bn}(t) - \Theta_{bc}(t)} \right)  \\
 & \qquad \qquad + \log \big( \Theta'_{sc}(T)/\Theta'_{bc}(T) \big) - \log \big( \Theta'_{sc}(0)/\Theta'_{bc}(0) \big),
\end{aligned}
\end{equation}
for all $T$.
\end{cor}

The decompositions \eqref{relValueEq1App}, \eqref{relValueEq2App}, and \eqref{relValueEq3App} from Theorem \ref{relValueThmApp} and Corollary \ref{relValueCorApp} extend the results of Theorem \ref{relValueThm} and Corollary \ref{relValueCor} to a setup that includes asset entry and exit over time. The most significant difference between \eqref{relValueEq1} and \eqref{relValueEq1App} is the local time process $d\L^n$, which enters into \eqref{relValueEq1App} negatively and measures the impact of asset exit from the modified small and market portfolios. In particular, the difference $d\L^n(\frac{1}{\Theta_{bn}} - \frac{1}{\Theta_{bn} - \Theta_{bc}})$, which is always negative, measures the impact of asset exit on the modified small portfolio $s'_c$ relative to the modified market portfolio $m'$. Thus, by comparing this term to the non-negative local time process $d\L^c$, we are able to compare the effects of asset exit and rank crossovers from bottom-ranks to top-ranks on cumulative relative portfolio returns. 

Similar conclusions emerge when comparing \eqref{relValueEq2App} and \eqref{relValueEq3App} to \eqref{relValueEq2} and \eqref{relValueEq3}. The local time process $d\L^n$ appears in both \eqref{relValueEq2App} and \eqref{relValueEq3App}, but not in either \eqref{relValueEq2} and \eqref{relValueEq3}. As with \eqref{relValueEq1App}, this local time process measures the effect of asset exit over time. In the case of the decomposition \eqref{relValueEq2App}, the process $d\L^n$ measures the effect of asset exit on the modified market portfolio $m'$, while in the case of \eqref{relValueEq3App}, it measures the effect of asset exit on the modified small portfolio $s'_c$.

Theorem \ref{relValueThmApp} and Corollary \ref{relValueCorApp} show that asset entry and exit do not overturn the basic insights of our results in Theorem \ref{relValueThm} and Corollary \ref{relValueCor}. Indeed, rank crossovers and the relative prices of ranked asset subsets remain key determinants of the relative returns of portfolios of bottom- and top-ranked assets even after introducing entry/exit. As for the rank effect, the decompositions \eqref{relValueEq1App}-\eqref{relValueEq3App} show that portfolios of bottom-ranked assets will outperform portfolios of top-ranked assets over time if both the relative price of bottom-ranked assets is approximately stable and the intensity of asset exit is not too great. Therefore, the conclusions regarding market efficiency discussed in Section \ref{sec:discussion} remain valid in this extension, with the added caveat that asset entry and exit over time must not be so frequent as to dominate the effects of rank crossovers over time. Ultimately, the intensity of rank crossovers versus asset entry and exit are empirical questions that depend on the specific market being examined.

\end{spacing}

\pagebreak

\begin{table}[H]
\vspace{0pt}
\begin{center}
\setlength{\extrarowheight}{3pt}
\begin{tabular} {| l ||c|c|c|}

\hline

      Commodity  & Exchange           & Start      &   Average and Standard Deviation      \\
                          &   where Traded   & Date    &   of Log Price Changes                          \\

\hline

  Soybean Meal      &  CBOT        &  1/1969  & 0.034 (0.303) \\
  Soybean Oil          &  CBOT       &  1/1969  & 0.027 (0.289)  \\
  Soybeans             &  CBOT        &  1/1969  & 0.027 (0.261)   \\
  Wheat                   &  CBOT        &  1/1969  & 0.027 (0.292)  \\
  Corn                     &  CBOT        &  1/1970  & 0.025 (0.260)  \\
  Live Hogs             &  CME          &  1/1970  & 0.022 (0.330)  \\
  Live Cattle            &  CME          &  1/1971  & 0.028 (0.201)  \\
  Cotton                  &  NYBOT      &  1/1973  & 0.018 (0.288)  \\
  Orange Juice       &  CEC           &  1/1973  & 0.029 (0.305)  \\
  Platinum              &  NYMEX      &  1/1973  & 0.041 (0.278)  \\
  Silver                   &  COMEX      &  1/1973  & 0.046 (0.320)  \\
  Coffee                  &  CSC           &  1/1974  & 0.013 (0.360)  \\
  Lumber                &  CME           &  1/1974  & 0.035 (0.326)  \\
  Gold                     &  COMEX     &  1/1975  & 0.045 (0.204)  \\
  Oats                     &  CBOT        &  1/1975  & 0.009 (0.345)  \\
  Sugar                   &  CSC          &  1/1975  & -0.032 (0.408)  \\
  Wheat, K.C.         &  KCBT        &  1/1977  & 0.016 (0.251)  \\
  Feeder Cattle      &  CME          &  1/1978  & 0.028 (0.169)  \\
  Heating Oil          &  NYMEX     &  1/1980  & 0.024 (0.328)  \\
  Cocoa                 &  CSC           &  1/1981  & 0.008 (0.301)  \\
  Wheat, Minn.      &  MGE          &  1/1981  & 0.007 (0.233)  \\
  Palladium            &  NYMEX     &  1/1983  & 0.065 (0.326)  \\
  Crude Oil            &  NYMEX     &  1/1984  & 0.026 (0.354)  \\
  RBOB Gasoline  &  NYMEX     &  1/1985 & 0.034 (0.348)  \\
  Rough Rice         &  CBOT       &  1/1987  & 0.035 (0.277)  \\
  Copper                &  COMEX    &  1/1989  & 0.027 (0.256)  \\
  Natural Gas         &  NYMEX    &  1/1991  & 0.014 (0.515)  \\
  Milk                      &  CME         &  9/1997  & 0.016 (0.277)  \\
  Brent Crude Oil   &  ICE           &  8/2008  & -0.042 (0.332)  \\
  Brent Gasoil        &  ICE           &  8/2008  & -0.047 (0.287)  \\

\hline

\end{tabular}
\end{center}
\vspace{-5pt} \caption{List of commodity futures contracts along with the exchange where each commodity is traded, the date each commodity started trading, and the annualized average and standard deviation (in parentheses) of daily log price changes for each commodity.}
\label{commInfoTab}
\end{table}

\begin{table}[H]
\vspace{15pt}
\begin{center}
\setlength{\extrarowheight}{3pt}
\begin{tabular} {|c||c|c|c|}

\hline

     & Price-Weighted (Market)   & Small                 & Big \\
     & Portfolio                            & Portfolio             & Portfolio           \\

\hline

  1974-2018       &                          3.44\%    \hspace{8pt} (17.88)    &                           10.37\%   \hspace{3pt} (18.72)    &                             0.40\%     \hspace{8pt} (20.15)   \\
  1974-1980       &                          10.65\%  \hspace{3pt} (32.58)    &                           15.11\%   \hspace{3pt} (34.97)    &                             7.80\%     \hspace{8pt} (35.01)  \\
  1980-1990       &                          -3.19\%  \hspace{4pt} (15.96)     &                           8.17\%     \hspace{8pt} (16.34)    &                            -8.19\%     \hspace{4pt} (18.66)  \\
  1990-2000       &  \hspace{-10pt}  0.28\%    \hspace{8pt} (8.03)     &                           5.54\%     \hspace{8pt} (10.09)    &                             -2.21\%     \hspace{4pt} (10.15)  \\
  2000-2010       &                           7.99\%    \hspace{8pt} (16.15)   &                           15.89\%   \hspace{3pt} (14.93)    &                              4.98\%     \hspace{8pt} (19.29)  \\
  2010-2018       &                           1.68\%    \hspace{8pt} (12.99)   &                           7.85\%     \hspace{8pt} (14.88)    &                             -0.59\%     \hspace{4pt} (14.20)  \\

\hline

\end{tabular}
\end{center}
\vspace{-5pt} \caption{Annualized average and standard deviation (in parentheses) of monthly returns for price-weighted (market) portfolio, small (bottom-ranked), and big (top-ranked) portfolios, 1974-2018.}
\label{returnsTab}
\end{table}

\begin{table}[H]
\vspace{50pt}
\begin{center}
\setlength{\extrarowheight}{3pt}
\begin{tabular} {|c||cc|cc|}

\hline

     &  \multicolumn{2}{c|}{Small Portfolio relative to} & \multicolumn{2}{c|}{Small Portfolio relative to}  \\
     &  \multicolumn{2}{c|}{Price-Weighted (Market) Portfolio}             & \multicolumn{2}{c|}{Big Portfolio}           \\
     & Average (st. dev.) & Sharpe ratio       & Average (st. dev.) & Sharpe ratio    \\

\hline

  1974-2018       &                              6.93\%    \hspace{7pt} (12.74)  & 0.54           &    9.97\%   \hspace{7pt} (18.53)  & 0.54  \\
  1974-1980       &                              4.46\%    \hspace{7pt} (16.92)  & 0.26           &    7.30\%   \hspace{7pt} (26.41)  & 0.28 \\
  1980-1990       &                             11.35\%  \hspace{2pt} (13.78)   & 0.82           &    16.35\% \hspace{2pt} (19.66)  & 0.83 \\
  1990-2000       &   \hspace{-10pt}   5.26\%    \hspace{7pt} (8.55)     & 0.62           &    7.76\%   \hspace{7pt} (12.97)  & 0.60 \\
  2000-2010       &                             7.91\%    \hspace{7pt} (13.08)   & 0.60           &    10.91\% \hspace{2pt} (18.39)  & 0.59 \\
  2010-2018       &   \hspace{-10pt}   6.17\%    \hspace{7pt} (9.93)     & 0.62           &    8.44\%   \hspace{7pt} (13.89)  & 0.61 \\

\hline

\end{tabular}
\end{center}
\vspace{-5pt} \caption{Annualized average, standard deviation (in parentheses), and Sharpe ratio of monthly returns for small (bottom-ranked) portfolio relative to price-weighted (market) and big (top-ranked) portfolios, 1974-2018.}
\label{relReturnsTab}
\end{table}

\begin{table}[H]
\vspace{15pt}
\begin{center}
\setlength{\extrarowheight}{3pt}
\begin{tabular} {|c||c|c|}

\hline

     &  Ratio of Bottom-Ranked       & Ratio of Bottom-Ranked       \\
     &  Prices relative to All Prices      &  Prices relative to                   \\
     &                                                 & Top-Ranked Prices                   \\

\hline

  Counterfactual, 2018      &  0.01275       &   0.00365    \\
  Actual, 2018                    &  0.26860       &   0.36724    \\

\hline

\end{tabular}
\end{center}
\vspace{-5pt} \caption{Counterfactual relative price of bottom-ranked commodities in 2018 needed for the small (bottom-ranked) portfolio to not have outperformed the price-weighted (market) and big (top-ranked) portfolios from 1974-2018, reported together with the actual relative price in 2018.}
\label{confacTab}
\end{table}


  



\begin{table}[H]
\vspace{50pt}
\begin{center}
\setlength{\extrarowheight}{3pt}
\begin{tabular} {|c||c|c|}

\hline

     &  Value Anomaly                         &  Value Anomaly         \\

\hline

  Intercept                                        &   0.324*  \hspace{12pt} (0.145)                             &  0.477*  \hspace{12pt} (0.211)                           \\ [0.2cm]
  
  Small Portfolio relative to               &  \multirow{2}{*}{0.270*** \hspace{2pt} (0.028)}  &                                                                          \\ [-0.1cm]
  Price-Weighted (Market) Portfolio &                                                                            &                                                                           \\  [0.2cm]

  Small Portfolio relative to              &                                                                            &  \multirow{2}{*}{0.380*** \hspace{2pt} (0.040)}  \\ [-0.1cm]
  Big Portfolio                                  &                                                                            &                                                                              \\

\hline  
  
  Adjusted $R^2$                            &   0.15                                                                   &   0.14                                  \\

\hline

\end{tabular}
\end{center}
\vspace{-5pt} \caption{Regression results with monthly value anomaly excess returns as dependent variable and rank effect excess returns as independent variables, 1974-2018. Standard errors are reported in parentheses.}
\label{valueTab}
\end{table}

\newpage

\begin{figure}[H]
\begin{center}
\vspace{-20pt}
\hspace{-20pt}\scalebox{0.65}{ {\includegraphics{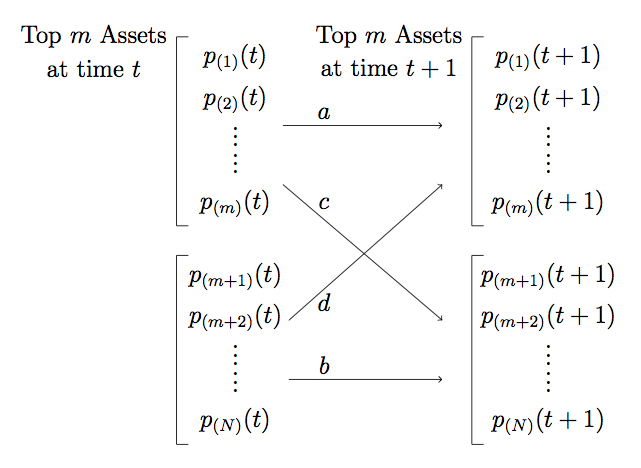}}}
\end{center}
\vspace{-24pt} \caption{A simple explanation of the rank effect.}
\label{simpleFig}
\end{figure}

\begin{figure}[H]
\begin{center}
\vspace{-4pt}
\hspace{-20pt}\scalebox{0.64}{ {\includegraphics{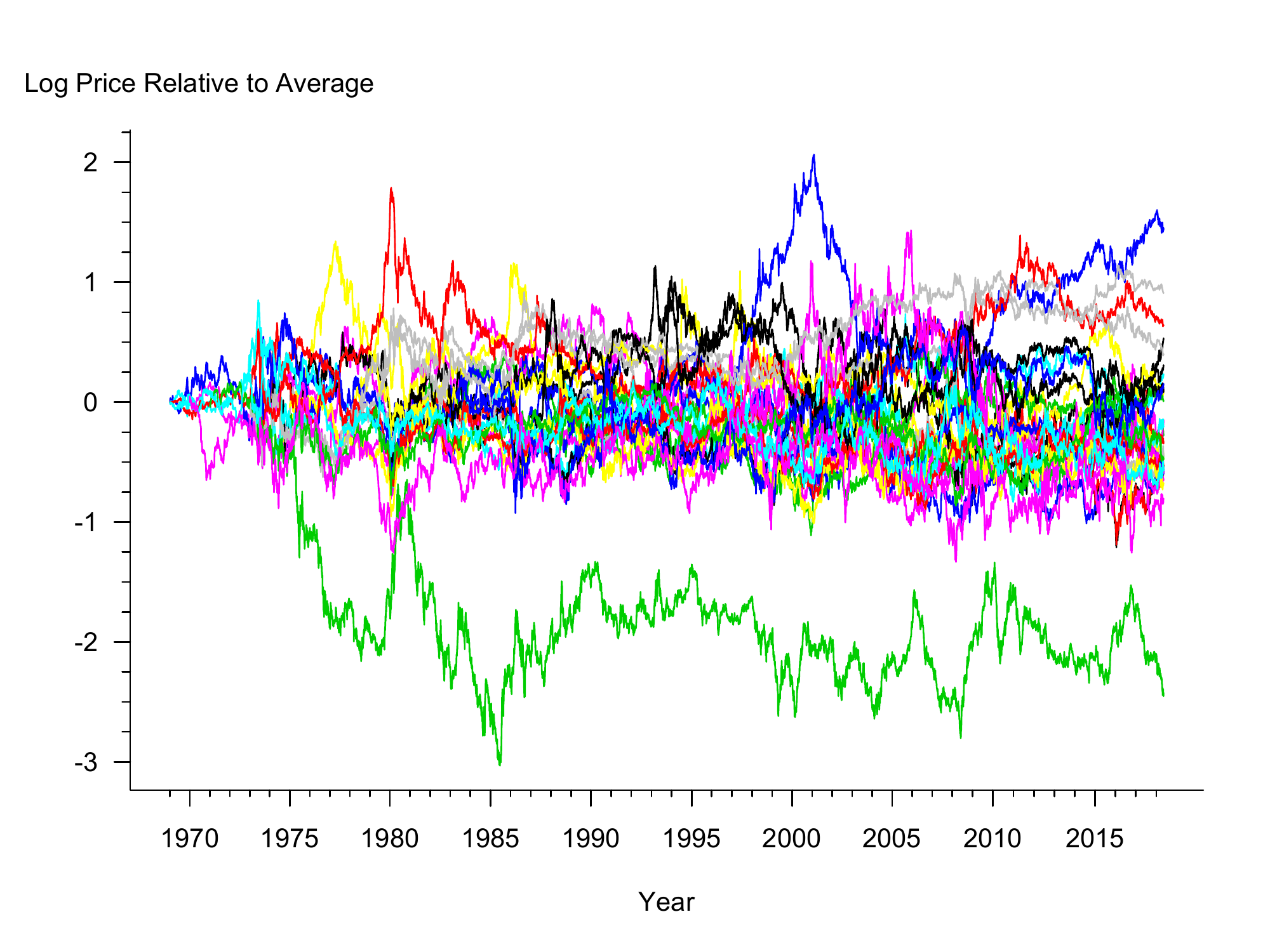}}}
\end{center}
\vspace{-24pt} \caption{Commodity prices relative to the average, 1969-2018.}
\label{relPricesFig}
\end{figure}

\begin{figure}[H]
\begin{center}
\vspace{-15pt}
\hspace{-20pt}\scalebox{0.63}{ {\includegraphics{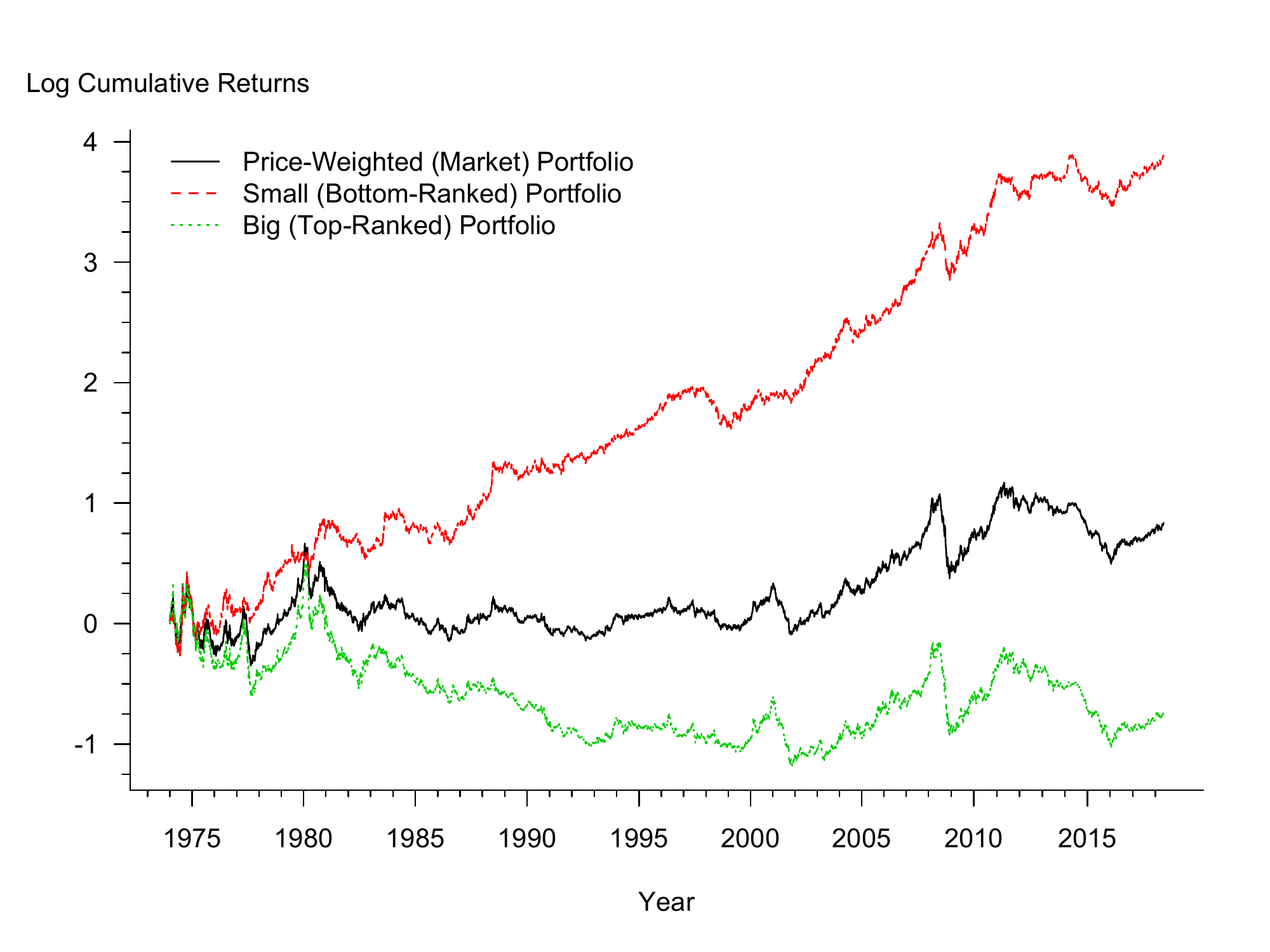}}}
\end{center}
\vspace{-24pt} \caption{Cumulative returns for price-weighted (market) portfolio, small (bottom-ranked), and big (top-ranked) portfolios, 1974-2018.}
\label{returnsFig}
\end{figure}

\begin{figure}[H]
\begin{center}
\vspace{-4pt}
\hspace{-20pt}\scalebox{0.63}{ {\includegraphics{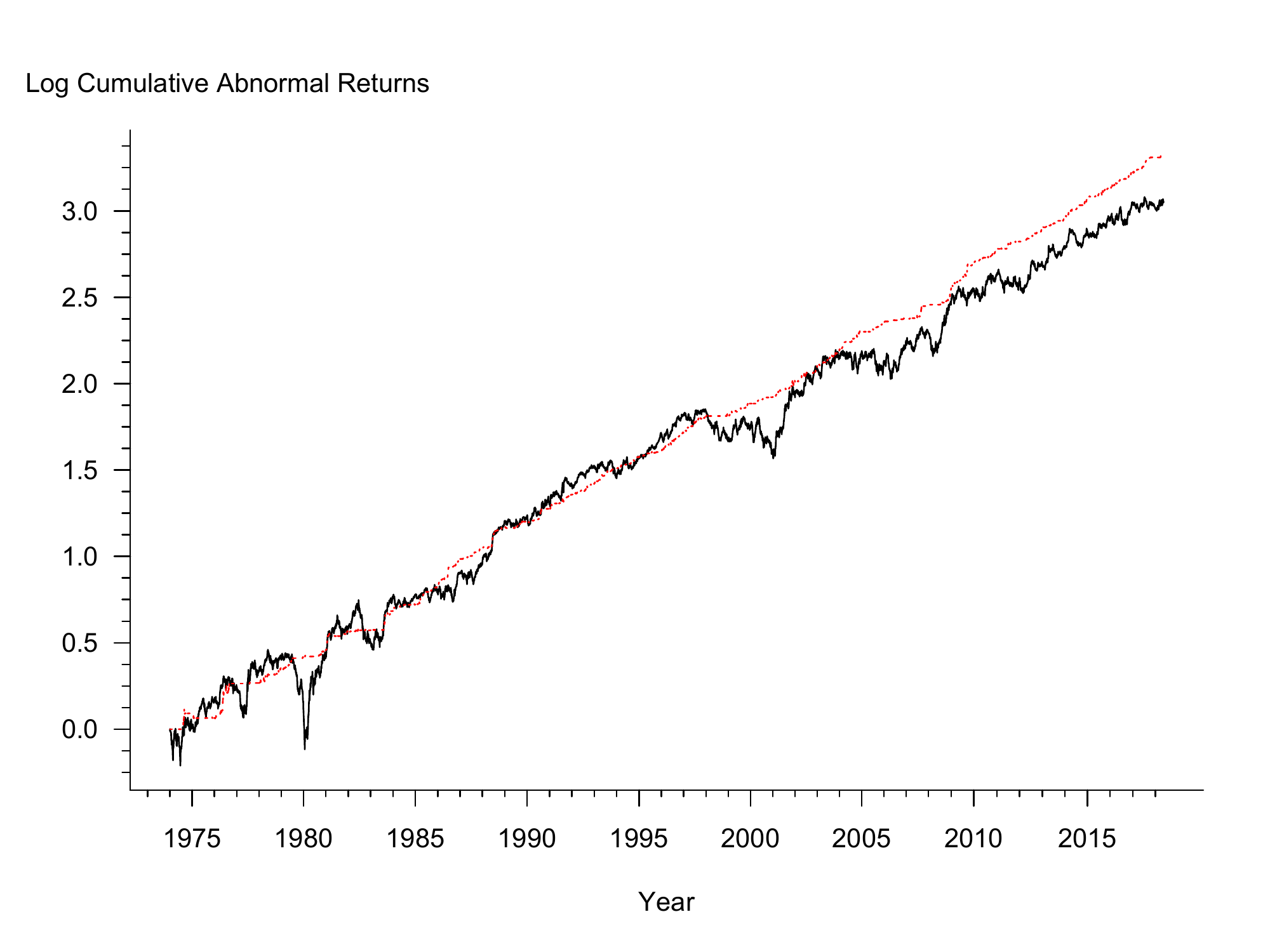}}}
\end{center}
\vspace{-24pt} \caption{Cumulative abnormal returns (solid black line) and rank crossovers (dashed red line) for small (bottom-ranked) portfolio, 1974-2018.}
\label{returnsSmallFig}
\end{figure}

\begin{figure}[H]
\begin{center}
\vspace{-15pt}
\hspace{-20pt}\scalebox{0.65}{ {\includegraphics{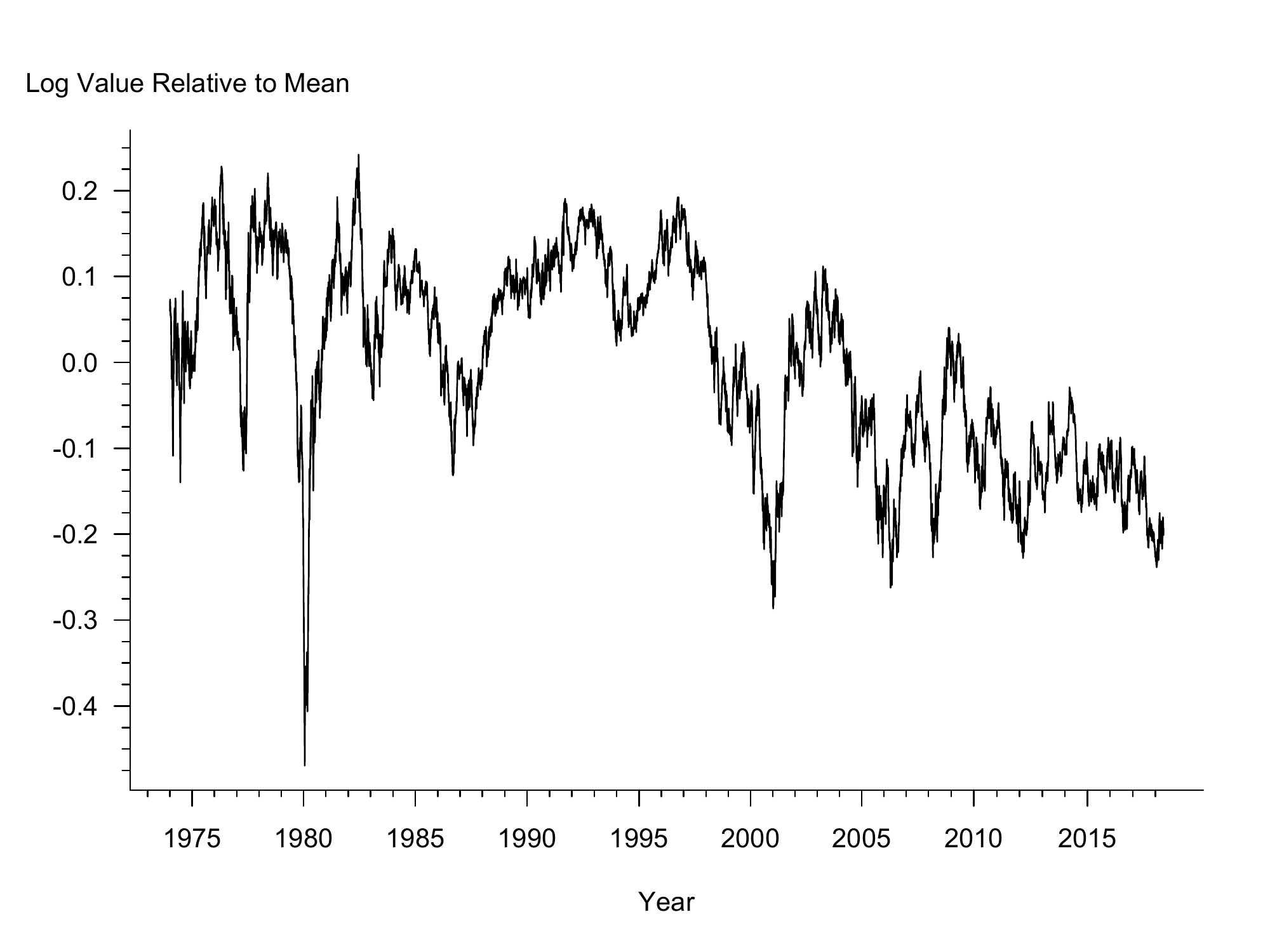}}}
\end{center}
\vspace{-24pt} \caption{Relative price of small (bottom-ranked) commodity futures, 1974-2018.}
\label{relPSmallFig}
\end{figure}

\begin{figure}[H]
\begin{center}
\vspace{-4pt}
\hspace{-20pt}\scalebox{0.65}{ {\includegraphics{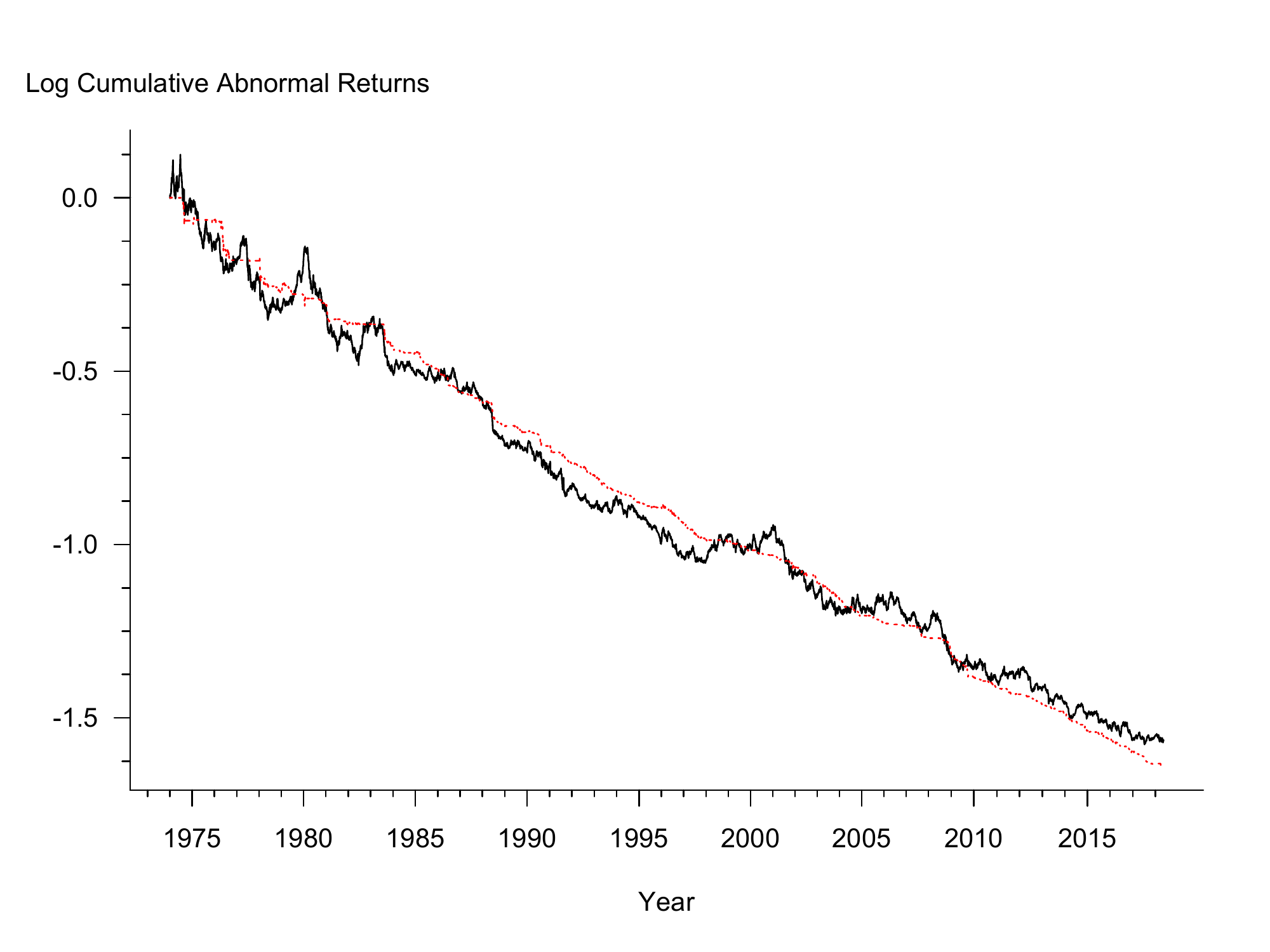}}}
\end{center}
\vspace{-24pt} \caption{Cumulative abnormal returns (solid black line) and rank crossovers (dashed red line) for big (top-ranked) portfolio, 1974-2018.}
\label{returnsBigFig}
\end{figure}

\begin{figure}[H]
\begin{center}
\vspace{-15pt}
\hspace{-20pt}\scalebox{0.65}{ {\includegraphics{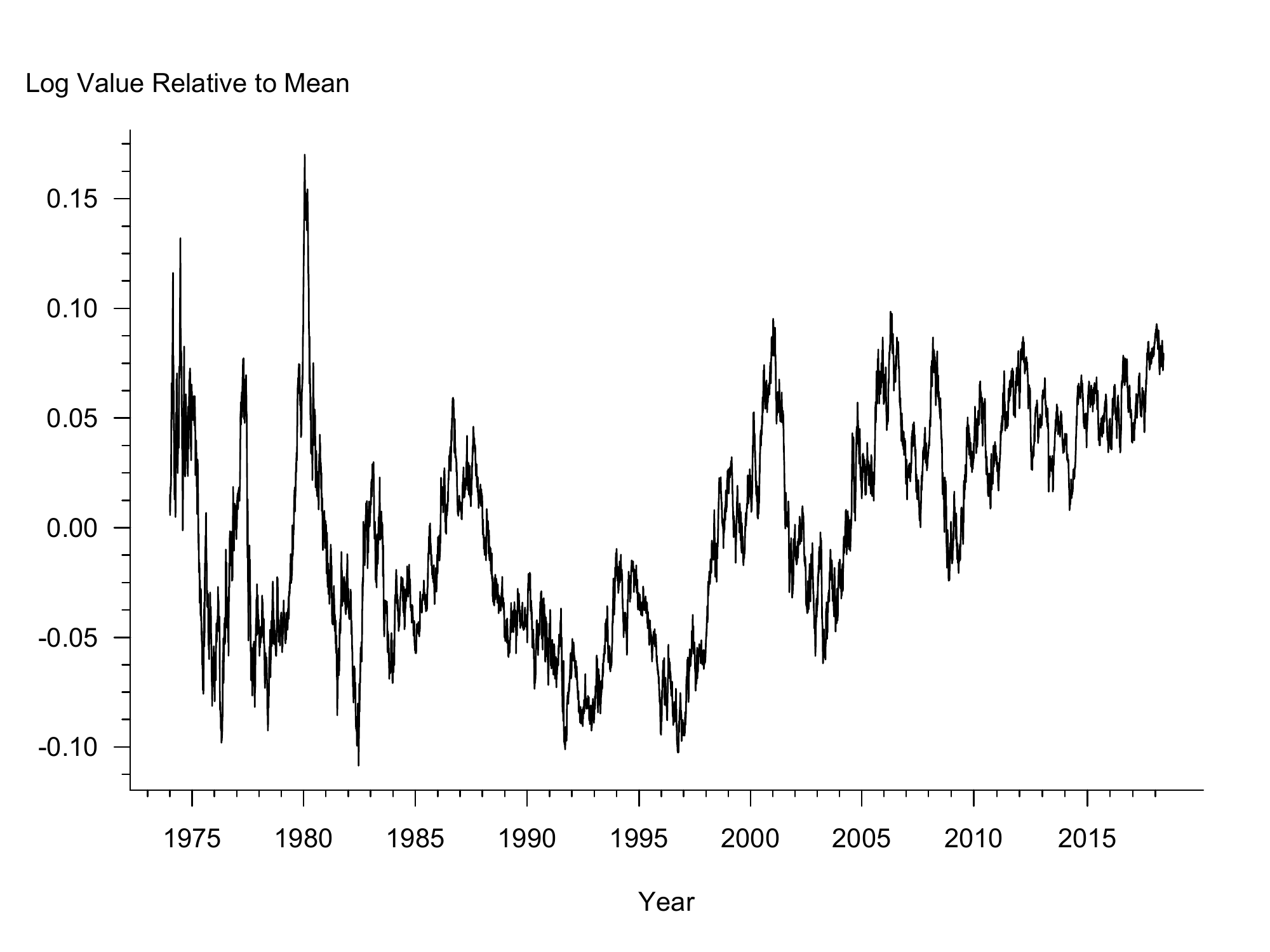}}}
\end{center}
\vspace{-24pt} \caption{Relative price of big (top-ranked) commodity futures, 1974-2018.}
\label{relPBigFig}
\end{figure}

\begin{figure}[H]
\begin{center}
\vspace{-4pt}
\hspace{-20pt}\scalebox{0.65}{ {\includegraphics{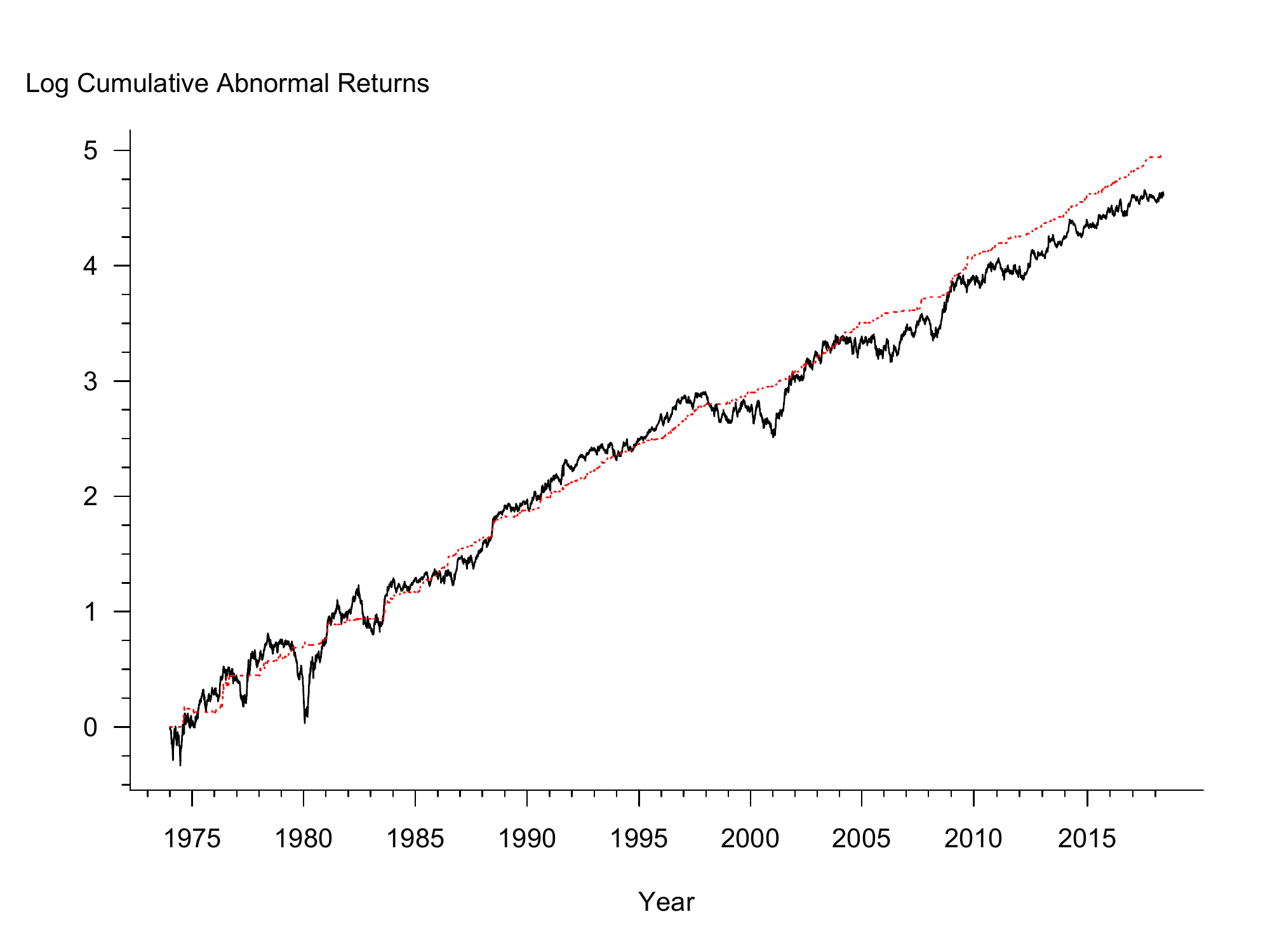}}}
\end{center}
\vspace{-24pt} \caption{Cumulative abnormal returns (solid black line) and rank crossovers (dashed red line) for the small portfolio relative to the big portfolio, 1974-2018.}
\label{returnsSmallBigFig}
\end{figure}

\begin{figure}[H]
\begin{center}
\vspace{-15pt}
\hspace{-20pt}\scalebox{0.62}{ {\includegraphics{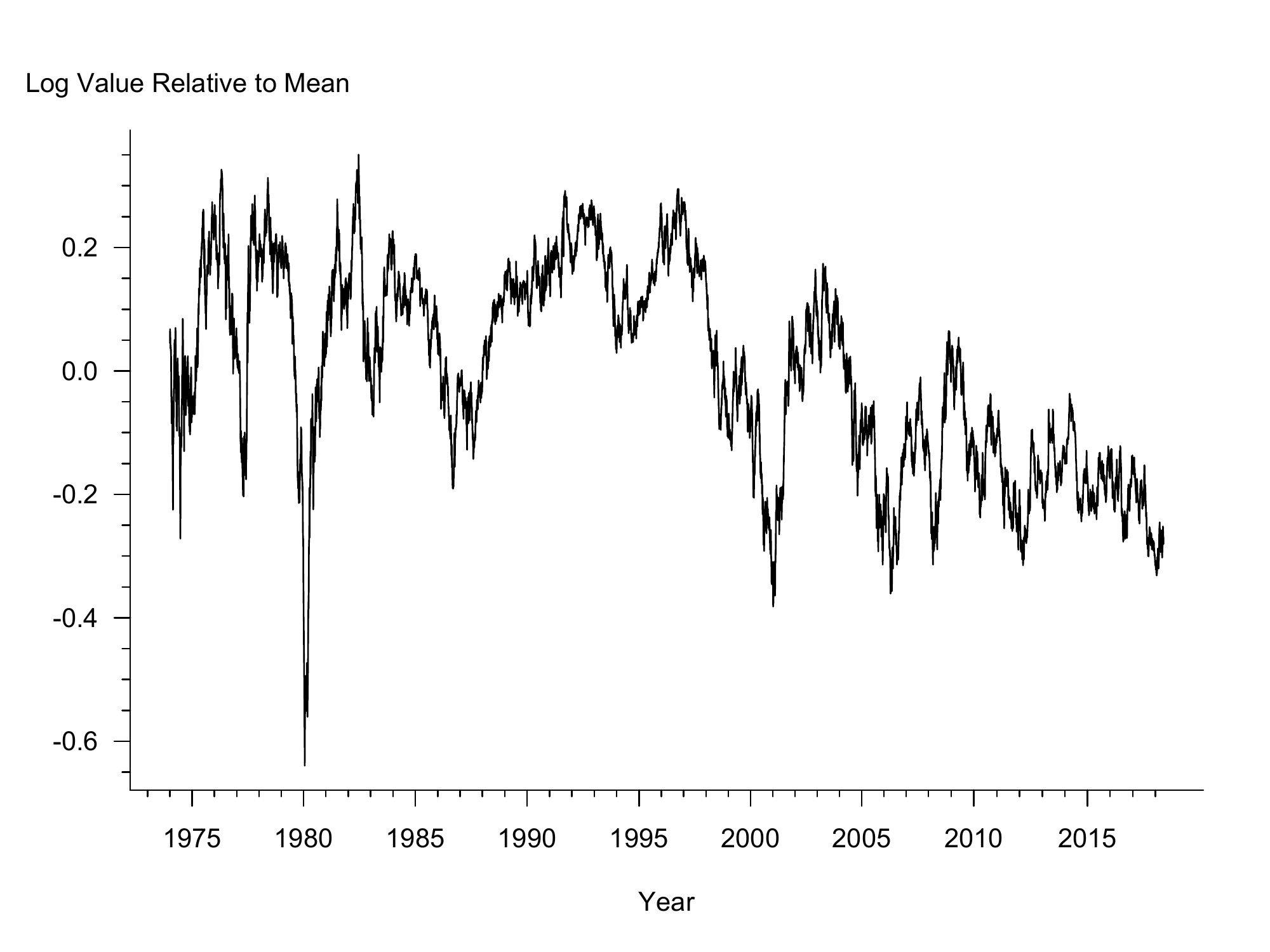}}}
\end{center}
\vspace{-24pt} \caption{Price of small (low-ranked) commodity futures relative to big (top-ranked) commodity futures, 1974-2018.}
\label{relPSmallBigFig}
\end{figure}

\begin{figure}[H]
\begin{center}
\vspace{-4pt}
\hspace{-20pt}\scalebox{0.62}{ {\includegraphics{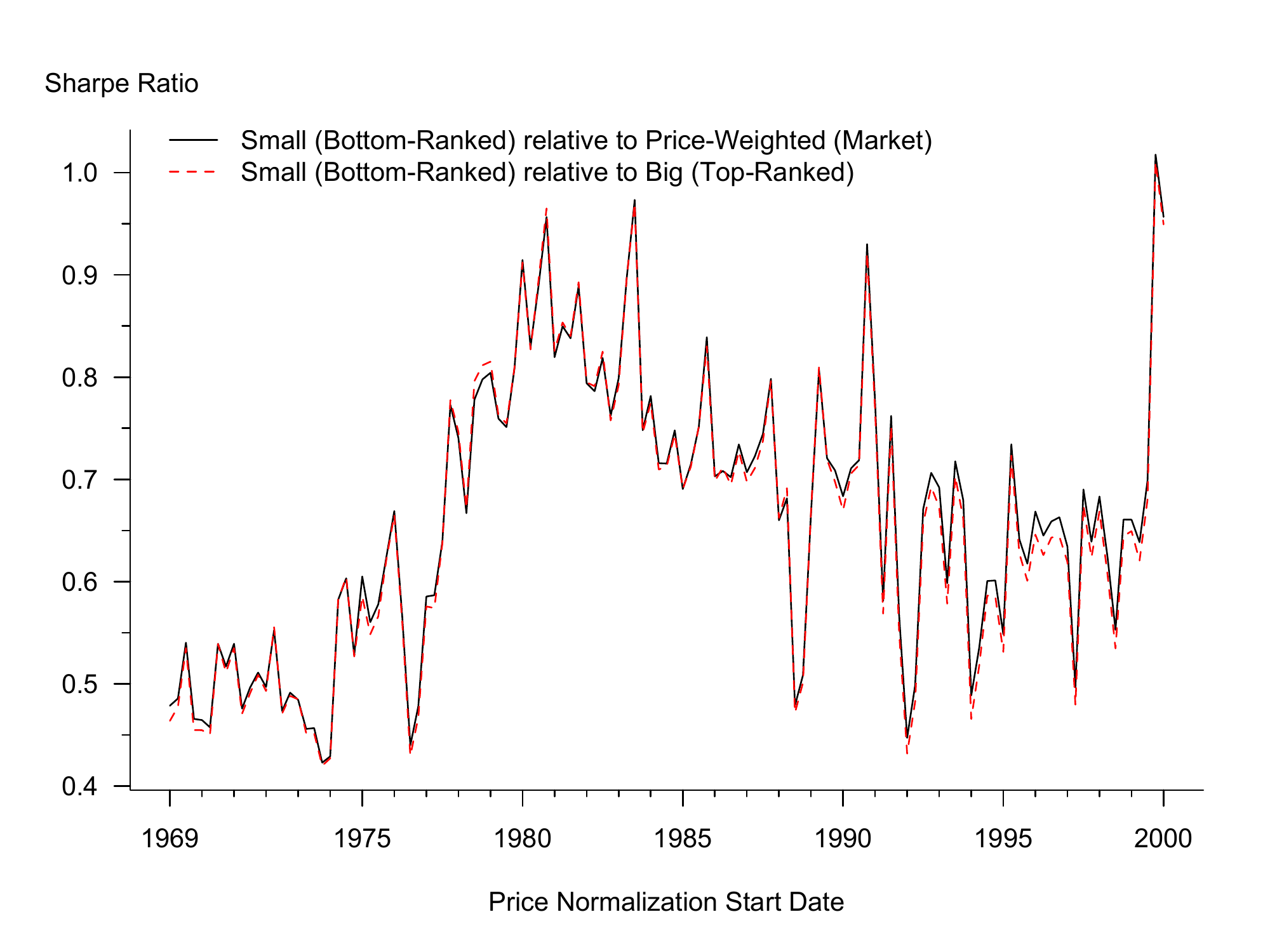}}}
\end{center}
\vspace{-24pt} \caption{Annualized sharpe ratio of monthly returns for small (bottom-ranked) portfolio relative to price-weighted (market) and big (top-ranked) portfolios for different price normalization start dates.}
\label{varyStartFig}
\end{figure}

\end{document}